\documentclass[11pt]{article}
\usepackage{amsmath}
\usepackage{amsthm}
\usepackage{amssymb}
\usepackage{url}
\usepackage[final]{showkeys}
\usepackage[dvipsnames]{xcolor}
\usepackage{graphicx}
\usepackage{bm}

\usepackage[export]{adjustbox}

\usepackage[vlined]{algorithm2e}

\usepackage{fullpage}

\usepackage{stmaryrd}

\usepackage[colorlinks=true, allcolors=blue]{hyperref}

\usepackage{latexsym}

\newtheorem{theorem}{Theorem}[section]
\newtheorem{lemma}[theorem]{Lemma}

\newtheorem{corollary}[theorem]{Corollary}
\newtheorem{definition}[theorem]{Definition}

\newcommand{\EQ}{\;=\;}
\newcommand{\LE}{\;\le\;}
\newcommand{\LT}{\;<\;}

\newcommand{\RR}{\mathbb{R}}

\newcommand{\Solve}{\mbox{\sc Solve}}
\newcommand{\Recover}{\mbox{\sc Recover}}
\newcommand{\FindOther}{\mbox{\sc FindOther}}
\newcommand{\ZeroRecover}{\mbox{\sc ZeroRecover}}

\title{\null\vspace*{-30pt}Strategy recovery for Mean-Payoff Games\\ is as hard as solving them from scratch}

\newcommand{\ignore}[1]{}
\newcommand{\remove}[1]{}
\newcommand{\reals}{\hbox{$\rlap{\rm I} \> \kern-.2mm{\rm R}$}}

\newcommand{\MPGG}{G=(V,V_0,V_1,E,c)}

\newcommand{\MAX}{{player~$1$}}
\newcommand{\MIN}{{player~$0$}}

\newcommand{\bG}{\bar{G}}
\newcommand{\bc}{\bar{c}}
\newcommand{\bV}{\bar{V}}
\newcommand{\bE}{\bar{E}}
\newcommand{\bs}{\bar{\sigma}}

\newcommand{\hs}{\hat{\sigma}}
\newcommand{\hG}{\hat{G}}

\newcommand{\TO}{\,{\to}\,}

\begin{document}

\author{Dani Dorfman \thanks{Max Planck Institute for Informatics, Saarbr\"ucken, Germany. E-mail: {\tt
    ddorfman@mpi-inf.mpg.de}.} \and Uri Zwick
\thanks{Blavatnik School of
Computer Science and AI, Tel Aviv University,
  Israel. E-mail: {\tt
    zwick@tau.ac.il}.}}

\date{}

\maketitle

\null\vspace*{-20pt}
\begin{abstract}
\setlength{\baselineskip}{12pt} 
\setlength{\parskip}{10pt} 
\noindent
Mean-Payoff Games (MPGs) form an interesting family of deterministic, infinite-duration, zero-sum, turn-based games played by two players on a finite weighted directed graph. The goal of the players is to minimize, respectively maximize, the \emph{limiting average} of the costs of the edges on an infinite path constructed by the two players.  
The complexity of solving MPGs is a long-standing open question. The decision problem associated with these games is known to be in $\text{NP}\cap \text{co-NP}$, and even $\text{UP}\cap \text{co-UP}$, yet the problem is not known to be in $\text{P}$. This is one of the simplest combinatorial problems with such a status. 

\noindent
A full solution of an MPG gives the \emph{value} of each vertex of the game, along with a pair of \emph{optimal strategies} for the two players. (Each player has a succinct \emph{positional} optimal strategy.) The problem of finding optimal strategies for the players given the values of all vertices is known as the \emph{strategy recovery} problem. We describe a simple reduction from the problem of solving MPGs to the strategy recovery problem for MPGs. This shows that the strategy recovery problem for MPGs is essentially as hard as solving MPGs from scratch, i.e., without knowing the values. The reduction uses simple \emph{game transformations} and simple \emph{self-reducibility} and \emph{peeling} arguments. This answers a question of Andersson and Miltersen. Mamino obtained a similar result for \emph{stochastic} games with a limiting average objective, but his technique does not extend to deterministic games such as MPGs. We go one step further and show that finding an optimal strategy for one of the players of an MPG remains hard even if an optimal strategy of the second player, and hence all the values, are known.
\end{abstract}

\section{Introduction}\label{S-intro}

Given the distances from a vertex~$s$ to all other vertices in a weighted directed graph, a tree of shortest paths from~$s$ can be constructed in linear time. Similarly, given the \emph{values} of all vertices, or states, of a \emph{discounted} \emph{Markov Decision Process (MDP)} \cite{Puterman94}, or of a discounted \emph{2-Player Turn-Based Stochastic Game (2-TBSG)} \cite{HaMiZw13}, the 2-player version of MDPs, corresponding optimal strategies of the two players can again be found in linear time. (It is not known, however, how to find the values of a 2-TBSG in polynomial time when the discount factor is part of the input.)

In sharp contrast, we show that for the family of \emph{Mean Payoff Games (MPGs)} \cite{EhMy79}, a 2-player deterministic version of 2-TBSGs with a \emph{limiting average} objective, instead of a discounted objective, the \emph{Strategy Recovery} problem, i.e., the problem of finding optimal strategies for the two players \emph{given} the values of all vertices, is, up to polynomial factors, as hard as solving MPGs from scratch.

More formally, we can define four problems related to the solution of MPGs. (See Figure~\ref{fig:variants}.) The first problem asks for a complete solution of an MPG, i.e., obtaining the \emph{values} of all vertices \emph{and} a pair of \emph{optimal strategies} for the two players. (The values can actually be computed in polynomial time given a pair of optimal strategies, or even just one optimal strategy.) The \emph{Strategy Recovery} problem, introduced by Andersson and Miltersen \cite{AnMi09}, asks for the computation of optimal strategies for the two players \emph{given} the values of all vertices. We introduce two additional problems. The \emph{Find Other Strategy} problem asks for finding an optimal strategy of one of the players, say player~$1$, \emph{given} an optimal strategy of player~$0$. Finally, the \emph{Zero Strategy Recovery} problem asks for optimal strategies $\sigma_0$ and~$\sigma_1$ of the two players under the promise that all the values in the input game are~$0$. For brevity, we refer to the four problems as \Solve, $\Recover$, $\FindOther$ and $\ZeroRecover$.

Our main result is the following:

\begin{theorem}\label{T1}
    The four MPG problems $\Solve$, $\Recover$, $\FindOther$ and $\ZeroRecover$ are equivalent under polynomial time Turing reductions.
\end{theorem}

This, in particular, answers an open question left by Andersson and Miltersen \cite{AnMi09} and extends a result of Mamino \cite{Mamino17} that applies only to games that allow stochastic moves. The proof of Theorem~\ref{T1}, without introducing stochastic moves, requires some new ideas. 

For parity games, Friedmann~\cite{Friedmann2011} shows that recovering winning strategies, even when all vertices are known to be winning for a specified player, and the problem of solving the game from scratch are equivalent under polynomial time Turing reductions.

The equivalence of $\FindOther$ to the other problems is somewhat surprising. The problem of finding an optimal strategy of a player given an optimal strategy of the other player seems to be a problem that involves only one player, and such problems can usually be solved in polynomial time. Indeed, an \emph{optimal counter strategy} of a player, also known as a \emph{best response}, i.e., an optimal strategy for a player when the other player uses a fixed strategy, can be found in polynomial time. However, a best response $\sigma_1$ of player~$1$ to an optimal strategy $\sigma_0$ of player~$0$ is not necessarily an optimal strategy for player~$1$, as can be shown by simple examples. (See Section~\ref{S-examples}.)

The equivalence of \Solve\ to the other three problems is under Turing reductions, and not Karp many-one reductions. We describe a way of solving an MPG by making a polynomial number of calls, and not just one call, to an oracle that solves the \Recover\ problem, or the \FindOther\ problem, or the \ZeroRecover\ problem. Each call is on a problem whose size is polynomial in the size of the input game.

We additionally show that \Recover, \FindOther, and \ZeroRecover\ are equivalent under Karp many-one reductions.

The strategy recovery problem for MPGs becomes easy again when given not only the values of all vertices, but also corresponding \emph{biases}, also known as \emph{potentials}. (See the next section.)

\begin{figure}[t]
    \centering
    \includegraphics[width=0.9\linewidth]{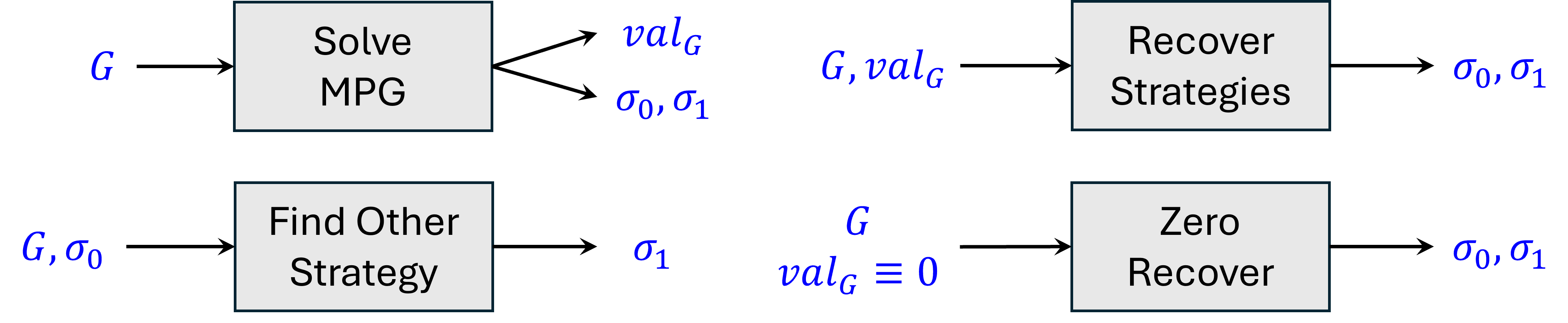}
    \caption{The four MPG problems considered. Here~$G$ stands for an MPG, $val_G$ is a vector giving the values of all vertices in~$G$, and $\sigma_0$ and $\sigma_1$ are optimal strategies of the two players. (The values $val_G$ are uniquely determined by~$G$, the optimal strategies $\sigma_0$ and $\sigma_1$ are not necessarily unique.)}
    \label{fig:variants}
\end{figure}

\subsection*{Paper organization}

In the next section we give a formal definition of MPGs and collect all the results that we rely on to obtain our reductions. In Section~\ref{S-overview} we give a high-level overview of our reductions. In Section~\ref{S-transform} we describe the three game transformations used by our reductions. The three transformations, in order, make all values of a game equal, negate all the values of a game, and take the average of two games. In Section~\ref{S-first} we describe a reduction from \Solve\ to \ZeroRecover, which is the hardest reduction. It uses the three transformations of Section~\ref{S-transform} along with \emph{self-reducibility} and \emph{peeling} techniques. In Section~\ref{S-second} we describe a reduction from \ZeroRecover\ to \FindOther\ that uses an \emph{escape} technique. The reductions from \FindOther\ to \Recover, and from \Recover\ to \Solve\ are trivial. In Section~\ref{S-ZeroRecover} we describe a Karp reduction from \Recover\ to \ZeroRecover, showing the equivalence of \Recover, \FindOther, and \ZeroRecover\ under polynomial time many-one reductions. In Section~\ref{S-examples} we give some illustrative examples. We conclude in Section~\ref{S-concl} with some concluding remarks and open problems.

\section{Mean-Payoff Games}\label{S-MPGs}

Mean-Payoff Games (MPGs) are a family of deterministic, infinite-duration, zero-sum, turn-based games played by two players on a finite weighted directed graph. They were defined, among others, by Ehrenfeucht and Mycielski \cite{EhMy79}, Gurvich, Karzanov, and Khachiyan \cite{GuKaKh88} and Zwick and Paterson~\cite{ZwPa96}. 

In this section we give a brief introduction to MPGs, and collect all the known facts about MPGs used by our reductions. For a more thorough treatment of MPGs, and for proofs of all the results stated here, see, e.g., Fijalkow et al.\ \cite{FijalkowEtAl26}.

\begin{definition}[MPGs]
    A \emph{Mean-Payoff Game} is a tuple $G=(V,V_0,V_1,E,c)$. Here, $V$ is the set of \emph{vertices} of the game, $V_0$ and $V_1$ form a partition of~$V$, i.e., $V_0\cup V_1=V$ and $V_0\cap V_1=\emptyset$, $E\subseteq V\times V$ is the set of directed \emph{edges} of the game, and $c:E\to\RR$ is a \emph{cost} function. It is assumed that each vertex in~$V$ has at least one outgoing edge. $V_0$ is the set of vertices controlled by player~$0$, the \emph{minimizer}, and $V_1$ is the set of vertices controlled by player~$1$, the \emph{maximizer}.
\end{definition}

\begin{definition}[Plays and outcomes]
The two players of an MPG jointly construct an infinite walk $u_0\TO u_1\TO u_2\dots$, also known as a \emph{play}, from an initial vertex $u_0$, as follows. If $u_i\in V_0$, then player~$0$ extends the walk by choosing an edge $(u_i\TO u_{i+1})\in E$. If $u_i\in V_1$, then player~$1$ chooses an edge $(u_i\TO u_{i+1})\in E$. The \emph{outcome} of the game is $\limsup_{N\to \infty}\frac{1}{N}\sum_{i=0}^{N-1}c(u_i\TO u_{i+1})$. \footnote{The outcome can also be defined as $\liminf_{N\to \infty}\frac{1}{N}\sum_{i=0}^{N-1}c(u_i\TO u_{i+1})$. This makes essentially no difference. The outcome cannot be defined as $\lim_{N\to \infty}\frac{1}{N}\sum_{i=0}^{N-1}c(u_i\TO u_{i+1})$ as this limit is not guaranteed to exist.} Player~$0$ wants to minimize the outcome, while player~$1$ wants to maximize it.
\end{definition}

\begin{definition}[Strategies]
    A strategy of player~$0$ in~$G$ is a function that maps every finite walk $u_0\TO u_1\TO\allowbreak \dots\TO u_i$ in~$G$, where $u_i\in V_0$, into an edge $(u_i\TO u_{i+1})\in E$. A strategy for player~$1$ is defined analogously. 
    We let $\Sigma_0$ and $\Sigma_1$ be the sets of all strategies of player~$0$ and player~$1$, respectively. A \emph{positional} strategy is a strategy that depends on $u_0\TO u_1\TO\dots\TO u_i$ only through~$u_i$. A positional strategy for player~$0$ is thus a function $\sigma_0$ that maps every vertex $u\in V_0$ to an edge $(u\TO v)\in E$. We let $\Sigma_0^P$ and $\Sigma_1^P$ be the set of all positional strategies of the two players. (Note that $\Sigma_0^P$ and $\Sigma_1^P$, unlike $\Sigma_0$ and $\Sigma_1$, are finite sets.)
\end{definition}

Strategies as defined above are \emph{deterministic}. It is possible to define \emph{randomized} strategies as probability distributions over deterministic strategies. But, as the games are deterministic and do not involve concurrent moves by the two players, it is not difficult to see that no player can benefit from using a randomized strategy. Thus, all the strategies considered in the paper are deterministic.

\begin{definition}[$play^{\sigma_0,\sigma_1}(u)$ and $out^{\sigma_0,\sigma_1}(u)$]
    If $\sigma_0$ and $\sigma_1$ are strategies of the two players, we let $play^{\sigma_0,\sigma_1}(u)$ be the walk obtained when the two players use $\sigma_0$ and $\sigma_1$ starting from $u\in V$. We let $out^{\sigma_0,\sigma_1}(u)$ be the resulting outcome.
\end{definition}

The following intuitive, but non-trivial, theorem says that MPGs are \emph{positionally determined}, i.e., they have \emph{values}, and that both players have positional strategies that are \emph{optimal} from every starting vertex. For a proof of the theorem, see, e.g., Ehrenfeucht and Mycielski \cite{EhMy79}, Gurvich, Karzanov, and Khachiyan \cite{GuKaKh88}, Bj{\"o}rklund, Sandberg, and Vorobyov \cite{BjSaVo04}, or Fijalkow et al.\ \cite{FijalkowEtAl26}. The theorem is a very restricted case of much more general theorems that apply to concurrent stochastic games. See, e.g., \cite{Everett1957,Gillette1957,liggett1969stochastic,mertens1981stochastic,Shapley53}.

\begin{theorem}\label{T-MPG}
    For every MPG $\MPGG$ there is a function $val_G:V\to \RR$ such that
    \[ val_G(u) \EQ \min_{\sigma_0\in \Sigma_0^P} \max_{\sigma_1\in \Sigma_1^P} \, out^{\sigma_0,\sigma_1}(u) \EQ \max_{\sigma_1\in \Sigma_1^P} \min_{\sigma_0\in \Sigma_0^P} \, out^{\sigma_0,\sigma_1}(u) \quad,\quad \text{for every $u\in V$}\;.\]    
    Furthermore, there exist positional strategies~$\sigma^*_0\in \Sigma_0^P$ and~$\sigma^*_1\in\Sigma_1^P$ of the two players such that
    \[ val_G(u) \EQ \sup_{\sigma_1\in \Sigma_1} \, out^{\sigma^*_0,\sigma_1}(u) \EQ \inf_{\sigma_0\in \Sigma_0} \, out^{\sigma_0,\sigma^*_1}(u) \quad,\quad \text{for every $u\in V$}\;,\] 
    (Note that the $\sup$ and $\inf$ are over \emph{arbitrary} strategies.) 
    The strategies $\sigma^*_0$ and~$\sigma^*_1$ are said to be positional \emph{optimal} strategies. (Positional optimal strategies are not necessarily unique.)
\end{theorem}

As Theorem~\ref{T-MPG} establishes the existence of positional optimal strategies, all strategies considered in the rest of the paper are positional. The following useful lemma is immediate:

\begin{lemma}\label{L-out}
    If $\MPGG$ is an MPG, $\sigma_0$ and $\sigma_1$ are positional strategies of the two players, and $u\in V$, then $play^{\sigma_0,\sigma_1}_G(u)$ is composed of a path that leads to a cycle which is repeated forever, and $out^{\sigma_0,\sigma_1}_G(u)$ is the average edge weight on the cycle. In particular, $out^{\sigma_0,\sigma_1}_G(u)$ can be computed in linear time.
\end{lemma}

As mentioned, solving an MPG $\MPGG$ means finding the values $val_G(u)$, for every $u\in V$, and a pair of positional strategies $\sigma_0$ and $\sigma_1$ that are optimal from every starting vertex. The corresponding decision problem is easily shown to be in $\text{NP}\cap \text{co-NP}$, and even $\text{UP}\cap \text{co-UP}$ (Jurdzi{\'n}ski \cite{Judzinski98}). Yet the problem is not known to be in $\text{P}$. This is one of the simplest combinatorial problems with such a status. The question whether MPGs can be solved in polynomial time is a major open problem.

The fastest known algorithms for solving MPGs are either \emph{pseudopolynomial}, i.e., run in polynomial time when the edge costs are integers encoded in \emph{unary} (Zwick and Paterson \cite{ZwPa96}, Brim, Chaloupka, Doyen, Gentilini, and Raskin \cite{brim2011faster}, Comin and Rizzi \cite{CominR17}), or randomized and \emph{subexponential} (Bj{\"o}rklund and Vorobyov \cite{BjVo05,BjVo07}, Hansen and Zwick \cite{HaZw15}, Zwick \cite{Zwick26}).

Given a strategy $\sigma_0$ of player~$0$ in an MPG $\MPGG$, it is possible to find in polynomial time an optimal \emph{counter strategy} $\sigma_1$ of player~$1$, and vice versa.

\begin{lemma}\label{L-counter}
    Given an MPG $\MPGG$ and a strategy~$\sigma_0$ of player~$0$, it is possible to find in polynomial time an optimal \emph{counter strategy} $\sigma_1$ of player~$1$, i.e., a strategy $\sigma_1$ such that $out^{\sigma_0,\sigma_1}_G(u)=\max_{\sigma'_1\in \Sigma^P_1} out^{\sigma_0,\sigma'_1}_G(u)$, for every $u\in V$. Similarly, given a strategy $\sigma_1$ of player~$1$, it is possible to find in polynomial time an optimal counter strategy $\sigma_0$ of player~$0$.
\end{lemma}

Lemma~\ref{L-counter} is proved by noting that finding an optimal counter strategy in an MPG corresponds to solving a 1-player MPG, i.e., an MPG in which $V_0=\emptyset$ or $V_1=\emptyset$, and such a problem is essentially equivalent to finding \emph{Minimum (Maximum) Mean Weight Cycles}, a problem that can be easily solved in polynomial time. (See, e.g., Karp \cite{Karp78}.) As corollaries of Lemma~\ref{L-counter} we get the following:

\begin{corollary}\label{C-values}
    If $\sigma_0$ is known to be an optimal strategy of player~$0$ in an MPG $\MPGG$, then the values $val_G(u)$, for $u\in V$, can be computed in polynomial time. The same holds if $\sigma_1$ is known to be an optimal strategy of player~$1$ in~$G$.
\end{corollary}

\begin{proof}
    Using Lemma~\ref{L-counter} we compute an optimal counter strategy $\sigma_{1}$ of player~$1$ to $\sigma_0$. We then have $val_G(u)=out_G^{\sigma_0,\sigma_1}(u)$, for every $u\in V$. The other case is proved similarly.
\end{proof}

\begin{corollary}\label{C-poly}
    Given an MPG $\MPGG$ and strategies $\sigma_0$ and $\sigma_1$ of the two players, it is possible to check in polynomial time whether both $\sigma_0$ and $\sigma_1$ are optimal strategies. If they are, all values can be computed in polynomial time.
\end{corollary}

\begin{proof}
    We use Lemma~\ref{L-counter} to compute optimal counter strategies $\sigma'_1$ and $\sigma'_0$ for $\sigma_0$ and $\sigma_1$. Using $\sigma'_0$ and $\sigma'_1$ we can check whether $\sigma_0$ and $\sigma_1$ are optimal counter strategies of each other. Indeed, $\sigma_1$ is an optimal counter strategy against $\sigma_0$ if and only if $out_G^{\sigma_0,\sigma_1}(u)=out_G^{\sigma_0,\sigma'_1}(u)$, for every $u\in V$. All these values can be computed in polynomial time using Lemma~\ref{L-out}. If $\sigma_0$ and $\sigma_1$ are optimal counter strategies of each other, then they are both optimal, and $val_G(u)=out_G^{\sigma_0,\sigma_1}(u)$, for every $u\in V$.
\end{proof}

We also need the following simple definition.

\begin{definition}[Restricted subgames]\label{D-sub}
    Let $\MPGG$ be an MPG and let $\emptyset\ne U\subseteq V$. Assume that each vertex $u\in U$ has at least one edge $(u\TO v)\in E$ such that $v\in U$. Then, we let $G[U]=(U,V_0\cap U,V_1\cap U, E\cap (U\times U),c')$ be the subgame of~$G$ \emph{restricted} to~$U$, or \emph{induced} by~$U$, where $c'$ is the function~$c$ restricted to $E\cap (U\times U)$.
\end{definition}

We also need the following very simple technical lemma. (It is actually a weaker form of Theorem~\ref{T-gain-bias} below.)

\begin{lemma}\label{L-equal} Let $G=(V,V_0,V_1,E,c)$ be an MPG. Then,
\[ val_G(u) \EQ 
    \begin{cases}
    \min_{(u\TO v)\in E}\, val_G(v) & \text{if } u\in V_0\\
    \max_{(u\TO v)\in E}\, val_G(v) & \text{if } u\in V_1
    \end{cases}\]
If $\sigma_0$ and $\sigma_1$ are positional optimal strategies of \MIN\ and \MAX, and $u\in V_i$, where $i\in\{0,1\}$, and $\sigma_i(u)=(u\TO v)$ then $val_G(u)=val_G(v)$.
\end{lemma}

Our reductions do not rely on the following theorem. (See, e.g., \cite{GuhaJKT16}.) It is useful, however, for understanding the general picture. We refer to the theorem in Sections~\ref{S-examples} and~\ref{S-concl}. The theorem uses the standard \emph{lexicographic} ordering on $\RR^2$: $(x,y)\le (x',y')$ if and only if $x<x'$ or $x=x'$ and $y\le y'$.

\begin{theorem}\label{T-gain-bias}
    Let $\MPGG$ be an MPG. Then, the following set of equations, known as the \emph{Lexicographic Bellman Equations}, in the unknowns $g(u)$, $b(u)$, for every $u\in V$, has a solution:
    \begin{align*}
    (\,g(u)\,,\,b(u)\,) &\EQ \min_{(u\TO v)\in E} (\,g(v)\,,\,c(u\TO v)-g(v)+b(v)\,) \quad,\quad u\in V_0 \;, \\
    (\,g(u)\,,\,b(u)\,) &\EQ \max_{(u\TO v)\in E} (\,g(v)\,,\,c(u\TO v)-g(v)+b(v)\,) \quad,\quad u\in V_1 \;.
    \end{align*}
    The $\min$ and $\max$ are taken with respect to the lexicographic order. For every solution $(g(u),b(u))$, for $u\in V$, of this set of equations we have:
    \begin{enumerate}
        \item $val_G(u)=g(u)$, for every $u\in V$.
        \item if $\sigma_0$ is a strategy of player~$0$ such that $\sigma_0(u)=(u\TO v)$ is one of the outgoing edges of~$u$ that achieves the minimum in the equation corresponding to~$u$, for every $u\in V_0$, then $\sigma_0$ is an optimal strategy for player~$0$ in~$G$.
        \item if $\sigma_1$ is a strategy of player~$1$ such that $\sigma_1(u)=(u\TO v)$ is one of the outgoing edges of~$u$ that achieves the maximum in the equation corresponding to~$u$, for every $u\in V_1$, then $\sigma_1$ is an optimal strategy for player~$1$ in~$G$.
    \end{enumerate}
\end{theorem}

The values $g(u)$ are known as \emph{gains}, while the values $b(u)$ are known as \emph{biases} or \emph{potentials}. Theorem~\ref{T-gain-bias} says that in any solution of the equations, the gains must equal the values of the vertices, and are thus uniquely determined. The biases, however, are not uniquely determined by the equations, even beyond the fact that an arbitrary constant can be added to all the biases. It is not known how to solve these non-linear equations in polynomial time.

Theorem~\ref{T-gain-bias} provides a very convenient way of proving that a pair of strategies $\sigma_0$ and~$\sigma_1$ are optimal for the two players; we simply need to provide appropriate gains and biases. We note however, that not all pairs of optimal strategies can be certified in this way. We may need separate biases for player~$0$ and player~$1$ to prove the optimality of an arbitrary pair of optimal strategies of the two players. (This requires a slight modification of Theorem~\ref{T-gain-bias} that we do not state here. We could, of course, always use Corollary~\ref{C-poly} that does not rely on Theorem~\ref{T-gain-bias}.)

It is also interesting to note that given both the values/gains \emph{and} the biases of all vertices in some solution of the Lexicographic Bellman equations, optimal strategies for both players can be easily recovered. Our reductions show that knowing only the values/gains is not enough, unless MPGs can be solved in polynomial time.

\section{Overview}\label{S-overview}

We prove Theorem~\ref{T1} using the following chain of reductions:
\[ \Solve \longrightarrow \ZeroRecover \longrightarrow \FindOther \longrightarrow \Recover \longrightarrow \Solve\]
The main challenge, to which most of the paper, and in particular Section~\ref{S-first}, is devoted to, is the reduction $\Solve \longrightarrow \ZeroRecover$. The reduction $\ZeroRecover \longrightarrow \FindOther$, described in Section~\ref{S-second}, is easier but also non-trivial. 

The reduction $\FindOther \longrightarrow \Recover$ is trivial. The input to \FindOther\ is a game~$G$ and an optimal strategy $\sigma_0$ of player~$0$. By Corollary~\ref{C-values}, we can use $\sigma_0$ to compute all values $val_G(u)$, for $u\in V$, in polynomial time. We can then call $\Recover(G,val_G)$ to obtain a pair of optimal strategies, and in particular an optimal strategy $\sigma_1$ of player~$1$.

The reduction $\Recover \longrightarrow \Solve$ is even more trivial. We simply ignore the given values and call $\Solve(G)$ to obtain a pair of optimal strategies.

A reduction from $\ZeroRecover$ to $\Recover$ is also trivial, as $\ZeroRecover$ is essentially a subproblem of $\Recover$. Thus, if we only want to show that $\Solve \longrightarrow \Recover$, we can use the chain $\Solve \longrightarrow \ZeroRecover \longrightarrow \Recover$, without having to rely on the reduction $\ZeroRecover \longrightarrow \FindOther$. A simple, but nontrivial, reduction in the opposite direction, i.e., from \Recover\ to \ZeroRecover, is given in Section~\ref{S-ZeroRecover}.

In the rest of this section we give an overview of the reduction $\Solve \longrightarrow \ZeroRecover$ which, as mentioned, is the most challenging reduction in the paper.

Suppose that $(\sigma_0,\sigma_1)\gets \ZeroRecover(G_0)$ is an oracle that given an MPG~$G_0$, along with a \emph{promise} that all values in~$G_0$ are~$0$, returns a pair $(\sigma_0,\sigma_1)$ of optimal strategies for the two players in~$G_0$. We show that using a polynomial number of calls to such an oracle we can solve an arbitrary MPG from scratch, i.e., implement an algorithm $(val_G,\sigma_0,\sigma_1)\gets \Solve(G)$ that given an MPG~$G$ returns the values of all its vertices and a pair of optimal strategies for the two players. We thus obtain a Turing reduction from \Solve\ to \ZeroRecover.

The reduction uses the following three transformations:
\begin{enumerate}
    \item Given an MPG $G$ we can construct in polynomial time an MPG $\bG$ such that the values of all vertices in~$\bG$ are equal to the maximal value $\mu_{\max}$ of a vertex in~$G$. This is done without knowing the values in~$G$. The game~$\bG$ is obtained by splitting each edge of~$G$ and introducing various \emph{jump} edges.
    \item Given an MPG $G$ we can construct in polynomial time an MPG $-G$, on the same vertex set as~$G$, such that $val_{-G}(v)=-val_G(v)$, for every vertex~$v$ of~$G$. (This construction is easy.)
    \item Given an MPG $G$ such that all vertices of~$G$ have value~$\mu$, and an MPG $H$ such that all vertices of~$H$ have value~$\nu$, we can construct an MPG $G\times H$ such that all vertices of $G\times H$ have value $\frac{\mu+\nu}{2}$. This is done without knowing $\mu$ and $\nu$. $G\times H$ is essentially the \emph{product} of $G$ and $H$, hence the notation $G\times H$. (The construction is actually slightly more general, see below.)
\end{enumerate}

The reduction is then composed of the following steps:
\begin{enumerate}
    \item Given an MPG $G$, use the first transformation to obtain an MPG $\bG$ in which the values of all vertices are equal to the maximal, and unknown, value $\mu_{\max}$ of a vertex in~$G$.
    \item Use the second transformation to obtain $-\bG$.
    \item Use the third transformation to obtain $\hG = \bG\times(-\bG)$. The crucial point is that all vertices of~$\hG$ have value~$0$.
    \item We can now call $(\hs_0,\hs_1)\gets \ZeroRecover(\hG)$ to obtain a pair of optimal strategies for $\hG$. We would like to use $\hs_0$ and $\hs_1$ to obtain optimal strategies $\bar{\sigma}_0$ and $\bar{\sigma}_1$ for the two players in~$\bG$. Unfortunately, it is not clear how to do it directly. Instead, we use a \emph{self-reducibility} approach that repeatedly calls $\ZeroRecover(\hG_i)$, where the $\hG_i$'s are subgames of $\hG$, i.e., games obtained by removing some of the edges of~$\hG$, until eventually finding the desired optimal strategies $\bar{\sigma}_0$ and $\bar{\sigma}_1$ for the two players in~$\bG$. (This is why we get a Turing reduction and not a Karp many-one reduction.)

    \item Our final challenge is to use the optimal strategies $\bar{\sigma}_0$ and $\bar{\sigma}_1$ for the two players in~$\bG$ to obtain optimal strategies $\sigma_0$ and $\sigma_1$ for the two players in the original game~$G$. This is again a non-trivial task. It is achieved using a \emph{peeling} process. The peeling process repeats the above four steps with various subgames of~$G$. The whole process only uses calls to $\ZeroRecover$.
\end{enumerate}

\section{Game transformations}\label{S-transform}

\subsection{Making all values equal}\label{sub-All-Equal}

\newcommand{\umax}{u_{\max}}

We begin with a simple transformation that takes an MPG $G=(V,V_0,V_1,E,c)$ with maximum value $\mu_{\max} = \max_{u\in V} val_G(u)$ and converts it into an MPG $\bG=(\bV,V_0,\bar{V}_1,\bar{E},\bc)$ such that $val_{\bG}(u)=\mu_{\max}$, for every $u\in \bV$, i.e., each vertex of $\bG$ has the maximal value~$\mu_{\max}$. The transformation does not need to know~$\mu_{\max}$. 

The game $\bG$ is obtained by replacing each edge $(u\TO v)\in E$ by two edges $u\TO uv$ and $uv\TO v$, where~$uv$ is a new vertex belonging to \MAX, as shown in Figure~\ref{fig:split}. We set $\bc(u\TO uv)=\bc(uv\TO v)=c(u\TO v)$. We let $\bV=V_0\cup \bV_1$, where $\bV_1=V_1\cup \{uv \mid (u\TO v)\in E\}$.
In other words, on every edge $(u\TO v)\in E$ we put an intermediate vertex $uv$ controlled by \MAX.
From each new vertex $uv$, we add edges $uv\TO w$ for every $w\in V\setminus\{v\}$, with $\bc(uv,w)=-M$, where $M$ is a sufficiently large number. (It is sufficient to let $M>4n\cdot\max_{(u,v)\in E} |c(u\TO v)|$, where $n=|V|$ is the number of vertices in the game.) In other words, from each such intermediate vertex we allow \MAX\ to jump to any original vertex of the graph. But, this jump has a very negative cost associated with it. The newly added edges are thus called \emph{jump} edges.

It is not difficult to see that in the new game $\bG$, the values of all vertices are equal to~$\mu_{\max}$, the maximal value of a vertex in~$G$. This is because \MAX\ can always choose to jump to a vertex~$u_{\max}$ such that $val_G(\umax)=\mu_{\max}$ and then follow an optimal strategy in~$G$. The cost of this initial jump is very negative, but as only one such jump is needed, it does not affect the value of the game. (This depends crucially on the fact that the limiting-average criterion is used.) The fact that the cost of jumping is so negative means that \MAX\ cannot benefit from jumping more than once. As all jumps are controlled by \MAX, all that \MIN\ can do is ensure that the value is at most~$\mu_{\max}$, by following her optimal strategy in the original game. We next formalize this argument.

\begin{figure}[t]
    \centering
    \includegraphics[width=0.2\linewidth,valign=c]{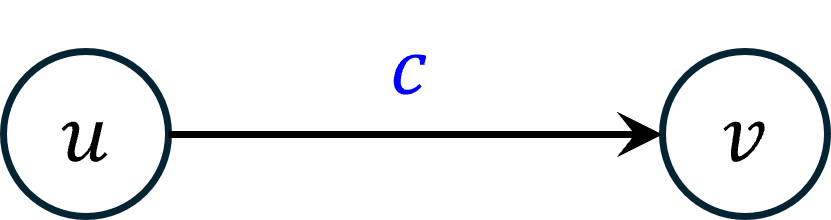}
    \hspace*{2cm}
    \includegraphics[width=0.2\linewidth,valign=c]{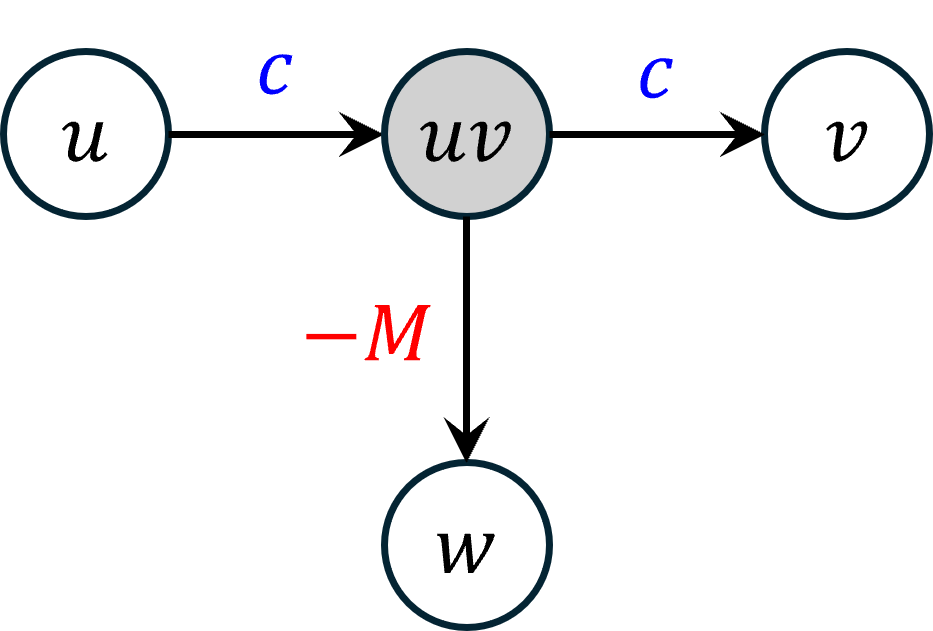}
    \caption{Splitting edges and introducing jump edges. Every edge $u\TO v$ of the original game~$G$ is split into two edges $u\TO uv$ and $uv\TO v$, where $uv$ is a newly introduced vertex belonging to \MAX. The costs of these two edges are set to $c(u\TO v)$. From $uv$ there are \emph{jump} edges of cost $-M$, where~$M$ is a sufficiently large number, to all original vertices of~$G$, except~$v$.}
    \label{fig:split}
\end{figure}

\begin{lemma}\label{L-equate}
Let $G=(V,V_0,V_1,E,c)$ be an MPG with maximum value $\mu_{\max} = \max_{u\in V} val_G(u)$ and let $\bG=(\bV,V_0,\bV_1,\bE,\bc)$ be the MPG constructed above. Then, $val_{\bG}(u)=\mu_{\max}$, for every $u\in \bV$.
\end{lemma}

\begin{proof}
  We begin by describing a strategy of \MAX\ that achieves a value of at least $\mu=\mu_{\max}$ from every vertex. The strategy is simple. Let $V_{\max}=\{ u\in V \mid val_G(u)=\mu \}$. Let $\sigma_1\in \Sigma^P_1$ be an optimal strategy of \MAX\ in~$G$.
  Let $v_{\max}\in V_{\max}$ be an arbitrary vertex.
  Extend $\sigma_1$ into a strategy~$\bs_1$ of \MAX\ in~$\bG$ in the following way: If $u\in V_1\subseteq \bV_1$ and $\sigma_1(u)=(u\TO v)$, then $\bs_1(u)=(u\TO uv)$. If $uv\in \bV_1\setminus V_1$, let
  \[ \bs_1(uv) \EQ \begin{cases}  uv\TO v & \text{ if } v\in V_{\max}\\
   uv\TO v_{\max} & \text{ if } v\not\in V_{\max}
  \end{cases}
  \]

  Let $\bs_0$ be an arbitrary strategy of \MIN\ in~$\bG$, and let $\sigma_0$ be the strategy of player~$0$ in~$G$ defined by $\sigma_0(u)=(u\TO v)$, if $\bs_0(u)=(u\TO uv)$, for every $u\in V_0$. We show that for every $u\in \bV$ we have $out^{\bs_0,\bs_1}_{\bG}(u)\ge \mu$. Suppose at first that $u\in V$ is also a vertex of~$G$. If $u\in V_{\max}$, then the infinite path obtained when \MIN\ and \MAX\ play according to $\bs_0$ and $\bs_1$ in~$\bG$ is identical to the play obtained when they play $\sigma_0$ and $\sigma_1$ in $G$, with each edge $u\TO v$ in~$G$ broken into the two new edges $u\TO uv$ and $uv\TO v$. This follows by Lemma~\ref{L-equal} as all original vertices encountered have value~$\mu$ and thus no jump edges are used. If $u\not\in V_{\max}$, $u\in V_i$, and $\bs_i(u)=(u\TO uv)$, where $i\in\{0,1\}$, then either $\bs_1(uv)=(uv\TO v)$ and $v\in V_{\max}$, or $\bs_1(uv)=(uv\TO v_{\max})$. In both cases, the play enters a vertex of~$V_{\max}$ and becomes equivalent to a play in~$G$ in which player~$1$ uses the optimal strategy~$\sigma_1$. Thus, the outcome in all these cases is at least~$\mu$. The same holds for all vertices $uv\in \bV_1\setminus V_1$.

  Let $\sigma_0$ be an optimal strategy of player~$0$ in~$G$. Let $\bs_0$ be the strategy of player~$0$ in~$\bG$ defined by $\bs_0(u)=(u\TO uv)$ whenever $\sigma_0(u)=(u\TO v)$, for every $u\in V_0$. Let $\bs_1$ be an arbitrary strategy of player~$1$ in~$\bG$. Let $u\in \bV$ and consider the cycle in $play_{\bG}^{\bs_0,\bs_1}(u)$. If the cycle contains a jump edge of weight $-M$, then the average weight of the cycle is smaller than any original edge weight in the graph. In particular, it is smaller than~$\mu$. Otherwise, the cycle corresponds to a cycle that can be formed in~$G$ when player~$0$ is using~$\sigma_0$, and its average weight is again at most $\mu$.

  Combining these two facts, we get the value of each vertex in~$\bG$ is exactly $\mu=\mu_{\max}$.  
\end{proof}

\subsection{Negating a game}\label{sub-Negating}

There is an obvious transformation that takes an MPG $\MPGG$ and converts it into an MPG $-G=(V,V_1,V_0,E,-c)$, i.e., a game in which the roles of the two players are switched and all costs are negated. The following lemma is then immediate:

\begin{lemma}\label{L-negate} Let $\MPGG$ be an MPG and let $-G=(V,V_1,V_0,E,-c)$. Then, $val_{-G}(u)=-val_G(u)$, for every $u\in V$.
\end{lemma}

\subsection{Averaging (or multiplying) two games}\label{sub-Averaging}

Given two MPGs $G$ and $H$, we construct an MPG $G\times H$ that essentially corresponds to playing~$G$ and~$H$ alternatingly and obtaining the average of the corresponding values. In the stochastic case, Mamino \cite{Mamino17} describes a simple way of achieving a similar effect by essentially taking the \emph{union} of~$G$ and~$H$, with some probabilistic transitions between them. In the deterministic case, we do not know how to achieve this using a union. Instead, we take the \emph{product} of the two games. The size of the resulting game is still polynomial.

\begin{figure}[t]
    \centering
    \includegraphics[width=0.2\linewidth]{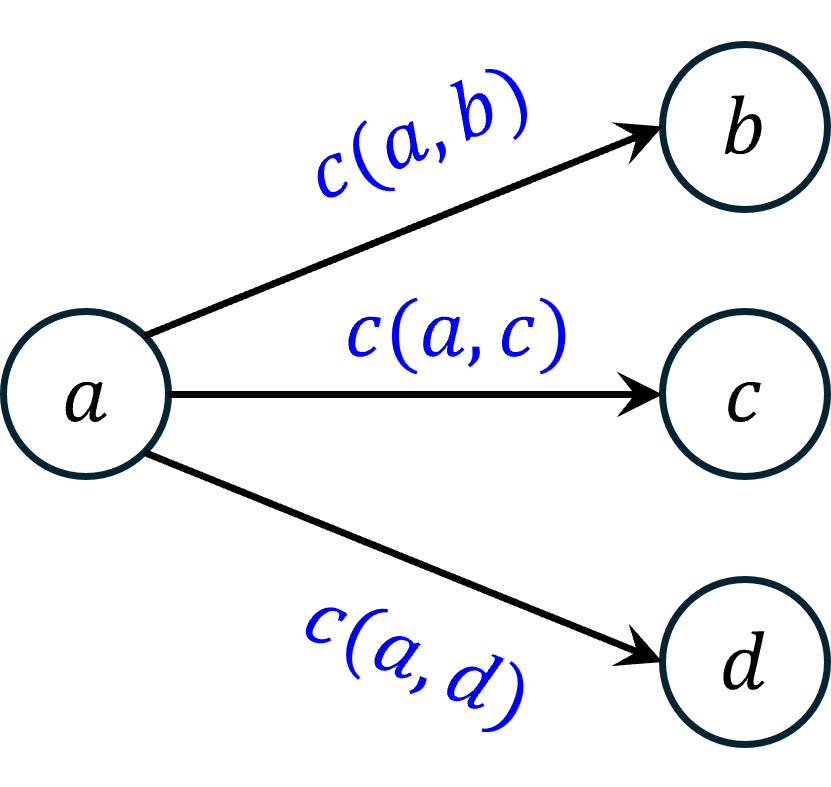}
    \hspace*{0.8cm}
    \includegraphics[width=0.2\linewidth]{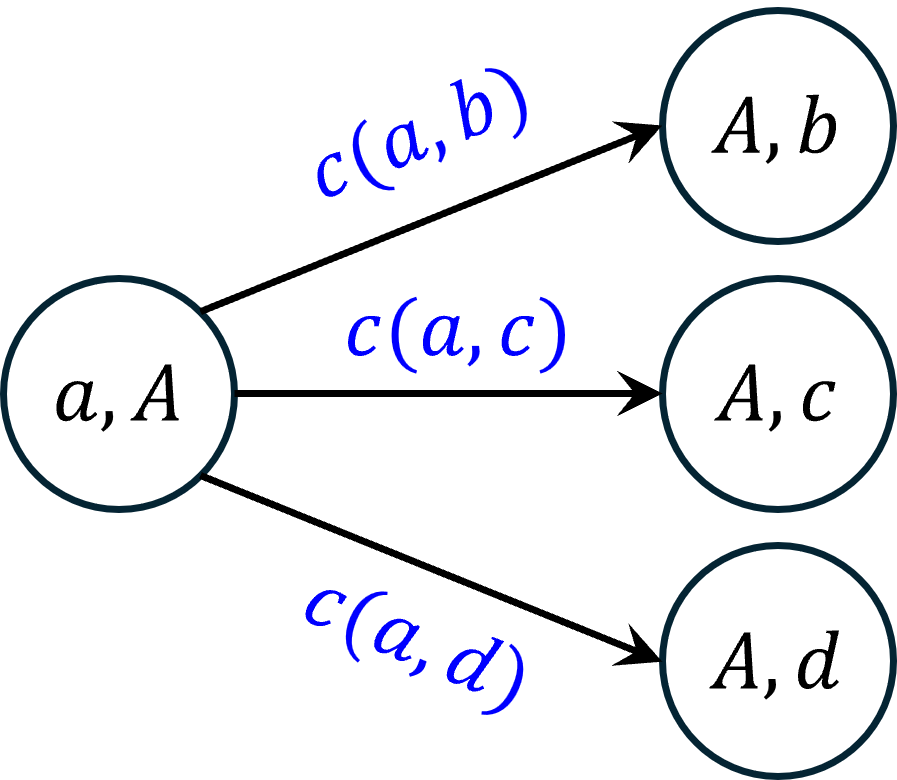}
    \hspace*{0.8cm}
    \includegraphics[width=0.2\linewidth]{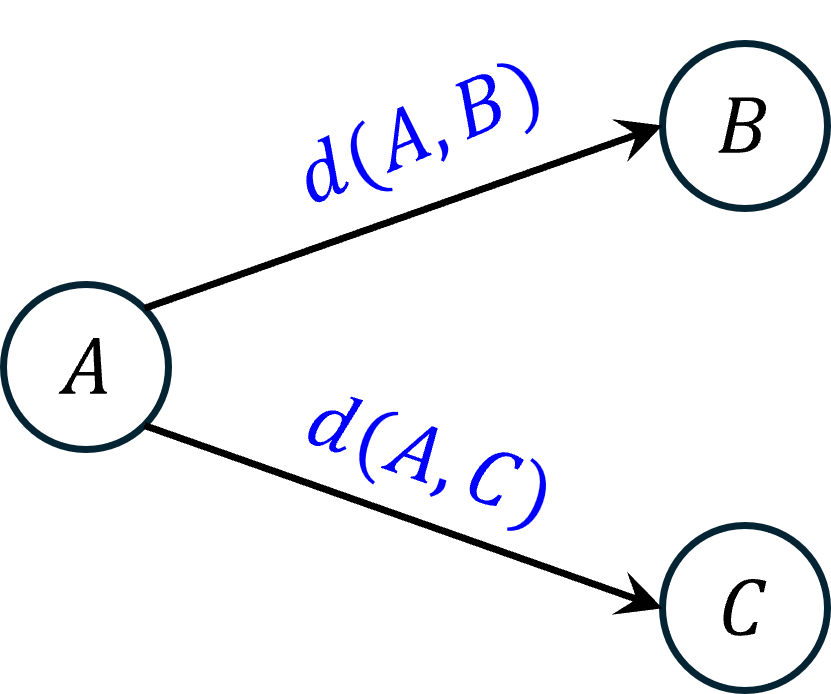}
    \hspace*{0.8cm}
    \includegraphics[width=0.2\linewidth]{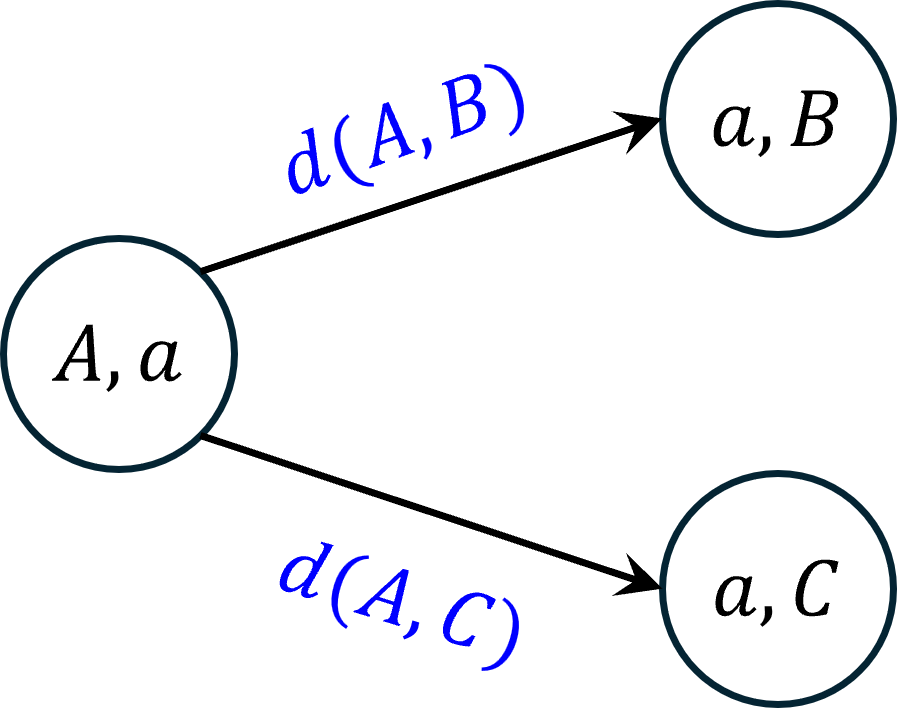}
    \caption{Forming the product of two games. If $a$ is a vertex of $G$ with three outgoing edges $a\TO b$, $a\TO c$ and $a\TO d$, and~$A$ is a vertex of~$H$, then $G\times H$ contains a vertex $(a,A)$ with outgoing edges to $(A,b)$, $(A,c)$ and $(A,d)$. The costs of these edges are identical to the costs of the original edges in~$G$. Vertex $(a,A)$ belongs to player~$0$ in $G\times H$ if and only if $a$ belongs to player~$0$ in~$G$. Similarly, if~$A$ is a vertex of~$H$ with two outgoing edges $A\TO B$ and $A\TO C$, and~$a$ is a vertex of~$G$, then $G\times H$ contains a vertex $(A,a)$ with outgoing edges to $(a,B)$ and $(a,C)$. The costs of these edges are identical to the costs of the original edges in~$H$.}
    \label{fig:product}
\end{figure}

In more detail. Let $G=(V,V_0,V_1,E,c)$ and $H=(U,U_0,U_1,F,d)$ be two MPGs. We assume that $V\cap U=\emptyset$. (If this does not hold, we take disjoint copies of~$G$ and~$H$.) The vertex set of $G\times H$ is $(V\times U)\cup (U\times V)$. A vertex in $G\times H$ thus corresponds to a vertex in~$G$ and a vertex in~$H$. A vertex in $V\times U$ indicates that a move in~$G$ should be made next, while a vertex in $U\times V$ indicates that a move in~$H$ should be made next. All edges in $G\times H$ belong to $(V\times U)\times (U\times V) \cup (U\times V)\times (V\times U)$. This ensures that moves in $G$ and in $H$ are performed alternatingly. If $v\TO w$ is an edge of~$G$, then $(v,u)\TO (u,w)$ is an edge of $G\times H$, for every $u\in U$. The weight of all these edges is equal to the weight of $v\TO w$ in~$G$. This corresponds to making a move in~$G$ and signaling that the next move should be in~$H$. Similarly, if $u\TO w$ is an edge of~$H$, then
$(u,v)\TO (v,w)$ is an edge of $G\times H$, for every $v\in V$. The weight of all these edges is equal to the weight of $u\TO w$ in~$H$. All edges in $G\times H$ are of this form. 

\newcommand{\CUP}{\;\cup\;}
\newcommand{\EF}{E\!F}
More formally, we let $G\times H=(W,W_0,W_1,\EF,cd)$ be the MPG defined as follows:
\[\begin{array}{c}
W \EQ (V\times U) \CUP (U\times V) \;, \\[3pt]
 W_0 \EQ (V_0\times U) \CUP (U_0\times V) \;, \\[3pt]
 W_1 \EQ (V_1\times U) \CUP (U_1\times V) \;, \\[3pt]
 \EF \EQ \{ (v,u)\TO (u,w) \mid (v\TO w) \in E \;,\; u\in U\} \CUP \{ (u,v)\TO (v,w) \mid (u\TO w) \in F \;,\; v\in V\} \;.
\end{array}\]
and
\[ cd((v,u)\TO (u,w)) \EQ c(v\TO w) \quad,\quad (v\TO w) \in E \;,\; u\in U \;,\]
\[ cd((u,v)\TO (v,w)) \EQ d(u\TO w) \quad,\quad (u\TO w) \in F \;,\; v\in V \;.\]

Alternatively, we could define the vertex set of $G\times H$ to be $W=V\times U\times\{0,1\}$, where the last component specifies whether the next step should be in~$G$ or in~$H$. We continue with the definition $W = (V\times U) \cup (U\times V)$, assuming that~$U$ and~$V$ are disjoint.

\begin{lemma}\label{L-product}
  Let $G=(V,V_0,V_1,E,c)$ and $H=(U,U_0,U_1,F,d)$ be two MPGs. Then, for every $v\in V$ and $u\in U$ we have
  \[val_{G{\times} H}(v,u) \EQ val_{G{\times} H}(u,v) \EQ\frac{val_{G}(v)+val_{H}(u)}{2} \;. \]
\end{lemma}

\newcommand{\st}{\sigma_0{\times}\tau_0}
\begin{proof} Let $\sigma_0$ and $\tau_0$ be optimal strategies of player~0 in $G$ and $H$, respectively. For every edge $u\TO v$, we let $head(u\TO v)=v$. Define a  strategy $\st$ of \MIN\ in $G\times H$ as follows:
\[ 
\begin{array}{ccl}
\st\,(v,u) \EQ ((v,u)\TO (u,head(\sigma_0(v)))) &\quad,\quad& \text{for every $v\in V_0$ and $u\in U$} \;, \\[4pt]
\st\,(u,v) \EQ ((u,v) \TO (v,head(\tau_0(u)))) &\quad,\quad& \text{for every $u\in U_0$ and $v\in V$} \;. 
\end{array}
\] 

Consider a play in~$G\times H$ in which \MIN\ uses $\st$ while \MAX\ uses an arbitrary positional strategy~$\rho_1$, not necessarily of a product form. It is easy to see that, for a fixed pair $(v,u)\in V\times U$, $play_{G{\times}H}^{\st,\rho_1}(v,u)$ corresponds to the interleaving of two plays in $G$ and $H$ in which \MIN\ uses $\sigma_0$ and $\tau_0$, respectively, while \MAX\ uses some not necessarily positional strategies $\sigma_1$ and $\tau_1$. As $\sigma_0$ and $\tau_0$ are optimal strategies for \MIN\ in~$G$ and~$H$, we get that for every $v\in V$ and $u\in U$, and for every strategy $\rho_1$ used by \MAX,
\[ out_{G{\times}H}^{\st,\rho_1}(v,u)
\EQ \frac{out_{G}^{\sigma_0,\sigma_1}(v)+out_{H}^{\tau_0,\tau_1}(u)}{2}
\LE \frac{val_{G}(v)+val_{H}(u)}{2} \;.  \]
We note that even though $\sigma_1$ and $\tau_1$ are not necessarily positional strategies, $play_{G}^{\sigma_0,\sigma_1}(v)$ and $play_{H}^{\tau_0,\tau_1}(u)$ are eventually periodic because $\st$ and $\rho_1$ are positional and hence $play_{G{\times}H}^{\st,\rho_1}(v,u)$ is eventually periodic. Hence, the $\limsup$'s in the definitions of $out_{G}^{\sigma_0,\sigma_1}(v)$ and $out_{H}^{\tau_0,\tau_1}(u)$ can be replaced by ordinary limits, justifying the equality above. In the same way we also get that
\[ out_{G{\times}H}^{\st,\rho_1}(u,v) \LE \frac{val_{G}(v)+val_{H}(u)}{2} \;.\] 

Similarly, from optimal positional strategies $\sigma_1$ and $\tau_1$ of \MAX\ in $G$ and~$H$, respectively, we can define a positional strategy $\sigma_1\times \tau_1$ for \MAX\ in $G\times H$ that guarantees inequalities in the opposite direction, and the lemma follows.
\end{proof}

The proof of Lemma~\ref{L-product} shows how to convert optimal positional strategies $\sigma_i$ and $\tau_i$ of player~$i$, where $i\in\{0,1\}$, in~$G$ and~$H$, respectively, into an optimal positional strategy $\sigma_i \times \tau_i$ of player~$i$ in $G\times H$. However, there may also exist optimal positional strategies of player~$i$ in $G\times H$ that are \emph{entangled}, i.e., not of this product form, as shown by the examples given in Section~\ref{S-examples}. Thus, an optimal positional strategy~$\rho_i$ of player~$i$ in $G\times H$ does not immediately give us optimal positional strategies for player~$i$ in~$G$ and $H$. To obtain optimal strategies for the two players in $G$ and in $H$, we need to use a self-reducibility approach described in Section~\ref{sub-self}.

\section{Reducing \Solve\ to \ZeroRecover}\label{S-first}

Let $\MPGG$ be an MPG. Let $\bG$ be the game obtained using the transformation of Section~\ref{sub-All-Equal}. By Lemma~\ref{L-equate}, the value of all vertices in~$\bG$ is~$\mu=\mu_{\max}$, the maximum value in~$G$. Let~$-\bG$ be the negation of $\bG$, obtained using the transformation of Section~\ref{sub-Negating}. By Lemma~\ref{L-negate}, the value of all vertices in $-\bG$ is $-\mu$. Finally, let $\hat{G}=\bG\times (-\bG)$, be the product of these two games, obtained using the transformation of Section~\ref{sub-Averaging}. (Since $\bG$ and $-\bG$ are defined on the same set of vertices, we take the product of disjoint copies of $\bG$ and $-\bG$.) The size of $\hat{G}$ is polynomial in the size of~$G$. By Lemma~\ref{L-product}, the value of all vertices in $\hat{G}$ is~$0$.

\subsection{The self-reducibility technique}\label{sub-self}

Suppose that all vertices of~$G$ have value~$\mu$ and that all vertices of~$H$ have value~$-\mu$. All vertices of $G\times H$ thus have value~$0$. The call $(\sigma_0,\sigma_1)\gets \ZeroRecover(G\times H)$ returns a pair of optimal strategies for $G\times H$. Our goal is to obtain a pair of optimal strategies for $G$, and possibly such a pair for~$H$. As explained in Section~\ref{sub-Averaging}, \MAX, for example, has an optimal strategy in $G\times H$ of a product form in which the edge chosen from a vertex $(v,u)$ of $G\times H$ depends only on~$v$. Such an optimal strategy yields immediately an optimal strategy for \MAX\ in~$G$. Unfortunately, as shown by an example given in Section~\ref{S-examples}, an optimal strategy of \MAX\ in $G\times H$ is not necessarily of this form. It is not immediately clear how to extract an optimal strategy for \MAX\ in~$G$ from a general optimal strategy for \MAX\ in $G\times H$.

We can nevertheless recover optimal positional strategies for the two players in the original games~$G$ and~$H$ by making a polynomial number of calls to $\ZeroRecover$. Suppose that $G=(V,V_0,V_1,E,c)$. 
We consider, one by one, the vertices of~$V_1$, i.e., the vertices of~$G$ controlled by \MAX, and for each one of them identify an outgoing edge that appears in an optimal positional strategy for \MAX\ in~$G$, together with previously selected outgoing edges from other vertices of~$V_1$.

Let $u\in V_1$, and let $u\TO v_1$, $u\TO v_2$, \ldots, $u\TO v_k$ be all the outgoing edges of~$u$ in~$G$. If $k=1$, then $u\TO v_1$ must of course be used in all optimal strategies for \MAX\ in~$G$. Assume, therefore, that $k>1$. For every $i\in [k]$, let~$G_i$ be the game obtained from~$G$ by leaving $u\TO v_i$ as the only outgoing edge of~$u$. If there is an optimal positional strategy for \MAX\ that uses the edge $u\TO v_i$, then the values of all vertices in~$G_i$ are~$\mu$, and vice versa. At least one index~$i\in[k]$ must satisfy this condition.

To check whether all values in~$G_i$ are~$\mu$, we form the product $G_i\times H$ and call $\ZeroRecover(G_i\times H)$. If all values in~$G_i$ are indeed~$\mu$, then all values in $G_i\times H$ are~$0$ and $\ZeroRecover(G_i\times H)$ returns a pair of optimal strategies of the two players in $G_i\times H$. By Corollary~\ref{C-poly}, we can verify in polynomial time that the two strategies returned are indeed optimal strategies and that the values of all vertices are~$0$. If this happens, we know that $u\TO v_i$ can be used by an optimal strategy of \MAX\ in~$G$.

What happens, however, if the values of all vertices in~$G_i$ are not all~$\mu$? In this case, not all the values in $G_i\times H$ are~$0$ and we have absolutely no guarantees on the behavior of $\ZeroRecover(G_i\times H)$ as the promise that all values in $G_i\times H$ are~$0$ does not hold. The call may return a pair of strategies in $G_i\times H$ that are not necessarily optimal, the output may not consist of valid strategies for the two players, the call may abort due to some internal error, or it may actually run forever.

To overcome this difficulty, we run all the calls $\ZeroRecover(G_1\times H) ,\ZeroRecover(G_2\times H),\ldots, \allowbreak\ZeroRecover(G_k\times H)$ in \emph{parallel}, or more precisely, we simulate such a parallel run on our sequential machine, alternately performing operations in each one of these calls. Whenever one of these runs, say, $\ZeroRecover(G_i\times H)$, terminates, we check whether its output consists of a valid pair of strategies, and if so, we check, in polynomial time, whether they are optimal strategies for the two players in $G_i\times H$ and whether all values in the game are~$0$. If so, we choose $u\TO v_i$ as the outgoing edge of~$u$, and abort all the currently running processes. If not, we discard this output and keep running the surviving processes. As mentioned above, at least one of the~$k$ calls must eventually terminate and pass the test. 

When an appropriate outgoing edge $u\TO v_i$ is found, we let $G\gets G_i$ and repeat the process for the other vertices of~$V_1$ until we are left with an optimal positional strategy for \MAX\ in~$G$.

An optimal strategy for \MIN\ in the original game~$G$ (with which we began the subsection) can be found in a similar manner. We can also repeat the process and find optimal strategies for the two players in~$H$, but this is not necessary in our case as we originally have $H=-G$.

By Corollary~\ref{C-values}, from either recovered optimal strategy we can compute the common value $\mu_{\max}$ in polynomial time.

\subsection{The peeling technique}\label{sub-peeling}

Using the self-reducibility technique, we obtained the value~$\mu_{\max}$ of all the vertices in~$\bG$ and optimal strategies $\bs_0$ and $\bs_1$ for the two players in~$\bG$, where $\bG$ is the game obtained from the original game~$G$ by applying the equalizing transformation of Section~\ref{sub-All-Equal}. We still need to compute the values $val_G(u)$ of all vertices in the original game~$G$ and convert $\bs_0$ and $\bs_1$ into optimal strategies $\sigma_0$ and $\sigma_1$ for the two players in~$G$. We do this using a \emph{peeling} technique.

Let $\bs_1$ be an optimal strategy of \MAX\ in $\bG=(\bV,V_0,\bV_1,\bE,\bc)$. Let 
\[J \EQ \{uv \in \bar{V}_1 \mid \bs_1(uv)\ne (uv\TO v)\} \;,\]
be the set of vertices of \MAX\ in $\bG$ from which $\bs_1$ uses jump edges. Let
\[U \EQ Reach_{\bG,0}^{\bs_1}(J) \] be the set of vertices from which \MIN\ can force reaching some vertex of~$J$ when \MAX\ plays~$\bs_1$. The set~$U$ is the winning set of a reachability game that can be easily solved in polynomial time as follows. Initialize $U\gets J$. As long as there is a vertex $u\in V_0$ and an edge $(u\TO v)\in \bE$ such that $v\in U$, or a vertex $u\in \bV_1$ such that $\bs_1(u)=(u\TO v)$ and $v\in U$, add~$u$ to~$U$. With some care, $U$ can actually be found in linear time. 

Let $\bV_{\max}=\bV\setminus U$. We first claim that~$U$ cannot contain all original vertices of~$\bG$, and thus $\bV_{\max}$ must contain at least one original vertex of~$G$. We then claim that $\bs_1$ ensures a value of $\mu_{\max}$ from every vertex of $\bV_{\max}$ without leaving $\bV_{\max}$ and without using any jump edge.

\begin{lemma}
    Let $\bV_{\max}=\bV\setminus U$. Then, $\bV_{\max}\cap V\ne \emptyset$. In other words, $\bV_{\max}$ contains at least one original vertex of~$G$.
\end{lemma}

\begin{proof}
    Assume, for the sake of contradiction, that there is a strategy $\bs_0$ of \MIN\ that guarantees reaching a vertex of~$J$ from every vertex of~$V$ when \MAX\ plays $\bs_1$. For every $u\in V$, the cycle in $play_{\bG}^{\bs_0,\bs_1}(u)$ must contain a jump edge of weight $-M$. Recall that $M>4n C$, where $C=\max_{e\in E} |c(e)|$. If the length of the cycle is~$2k$, where $1\le k\le n$, then 
    \[ out_{\bG}^{\bs_0,\bs_1}(u) \LE \frac{(2k-1)C-M}{2k} \LT \frac{(2k-1)C-4nC}{2k} \LE -\frac{(2n+1)C}{2n} \LT -C \LE \mu_{\max} \;, \]
    contradicting the assumption that $\bs_1$ is an optimal strategy of player~$1$ in~$\bG$ that ensures a value of at least~$\mu_{\max}$ from each vertex. Thus $U\not\supseteq V$ and $\bV_{\max}\cap V\ne \emptyset$.
\end{proof}

\begin{lemma}\label{L-2}
    $\bs_1$ ensures a value of at least $\mu_{\max}$ from every vertex of $\bV_{\max}$ without leaving $\bV_{\max}$ and without using any jump edge.
\end{lemma}

\begin{proof}
    Any play which starts in $\bV_{\max}$ in which \MAX\ plays $\bs_1$, and \MIN\ plays arbitrarily, stays in~$\bV_{\max}$ and does not use any jump edge. Since $\bs_1$ is a strategy that guarantees a value of at least~$\mu_{\max}$ from every vertex of~$\bG$, the value of each such play is also at least $\mu_{\max}$.
\end{proof}

Let $\sigma_1$ be a (partial) strategy for \MAX\ in~$G$ defined as follows. For every $u\in V_1\cap \bV_{\max}$, if $\bs_1(u)=(u\TO uv)$, let $\sigma_1(u)=(u\TO v)$. The following corollary follows immediately from Lemma~\ref{L-2}.

\begin{corollary}
    $\sigma_1$ ensures a value of at least $\mu_{\max}$ from every vertex of $V_{\max}=V\cap \bV_{\max}$ of~$G$ without leaving $V_{\max}$.
\end{corollary}

We now know that the values of all vertices of $V_{\max}=V\cap \bV_{\max}\ne \emptyset$ is $\mu_{\max}$ and that $\sigma_1$ ensures at least this value from each vertex of~$V_{\max}$. (The values of the vertices in $V_{\max}$ cannot be larger than~$\mu_{\max}$, as $\mu_{\max}$ is assumed to be the largest value of a vertex in~$G$.) The set $V_{\max}$ does not necessarily contain all vertices of value~$\mu_{\max}$ in~$G$. (The optimal strategy $\bs_1$ may use jump edges from some vertices even if these jump edges can be avoided.) We can possibly enlarge the set of vertices of~$G$ known to have value~$\mu_{\max}$ by computing the closure, or the \emph{max-attractor}
\[ A \EQ Reach_{G,1}(V_{\max}) \;, \] 
the set of all vertices in~$G$ from which \MAX\ can force reaching a vertex of~$V_{\max}$, no matter which strategy is used by player~$0$. The set $A$ can again be found in polynomial time, and even linear time. 

We can easily extend $\sigma_1$ into a strategy $\alpha_1$ of \MAX\ that ensures a value of at least~$\mu_{\max}$ from all the vertices of~$A$. From a vertex $u\in A\setminus V_{\max}$, follow a strategy that ensures getting eventually to some vertex of~$V_{\max}$. From a vertex $u\in V_{\max}$, follow $\sigma_1$. (The set~$A$ is still not guaranteed to contain all vertices of~$G$ of value~$\mu_{\max}$.) We thus have:

\begin{lemma}\label{L-alpha-1}
    The strategy $\alpha_1$ of player~$1$ in~$G$ ensures a value of at least $\mu_{\max}$ from every vertex~of~$A$ while remaining in~$A$.
\end{lemma}

We can convert the optimal strategy $\bs_0$ of player~$0$ in~$\bG$ into a strategy $\alpha_0$ of player~$0$ in~$G$ as follows: If $\bs_0(u)=(u\TO uv)$, let $\alpha_0(u)=(u\TO v)$, for every $u\in V_0$. The following claim is also obvious:

\begin{lemma}\label{L-alpha-0}
    The strategy $\alpha_0$ of player~$0$ in~$G$ ensures a value of at most $\mu_{\max}$ from every vertex~of~$V$.
\end{lemma}

We made some progress, as we have discovered the maximum value~$\mu_{\max}$ and identified a nonempty set $A$ of vertices that all have the value $\mu_{\max}$. If $A\ne V$ we need to find the values of all vertices in $B=V\setminus A$. (Note that $A\cup B = V$ and $A\cap B=\emptyset$.) We also need to find optimal strategies for the two players that ensure all these values.

If $B\ne\emptyset$, we consider the game $G[B]$, i.e., the game~$G$ restricted to~$B$. (See Definition~\ref{D-sub}.) The game $G[B]$ is obtained by removing all vertices of~$A$ and their edges from~$G$. By the closure of~$A$, each vertex $u\in B$ has at least one outgoing edge in~$G[B]$, so~$G[B]$ is a valid game. We now claim:

\begin{lemma}\label{L-peel}
    For every $u\in B$, $val_G(u)=\min\{ \mu_{\max},val_{G[B]}(u)\}$.
\end{lemma}

\begin{proof}
    For every $u\in V$ we have $val_G(u)\le \mu_{\max}$. We next show that for every $u\in B$ we have $val_G(u)\le val_{G[B]}(u)$. Indeed, let $\sigma_0$ be an optimal strategy of player~$0$ in~$G[B]$. Extend it in an arbitrary way into a strategy of player~$0$ in~$G$. Let $\sigma_1$ be an arbitrary strategy of player~$1$ in~$G$. For every $u\in B$, $play_G^{\sigma_0,\sigma_1}(u)$ stays in~$B$. Hence $out_G^{\sigma_0,\sigma_1}(u)\le val_{G[B]}(u)$. As this holds for every $\sigma_1$, we have $val_G(u)\le val_{G[B]}(u)$ and thus $val_G(u)\le \min\{ \mu_{\max},val_{G[B]}(u)\}$.

    We next show the opposite inequality. Let $\sigma_1^B$ be an optimal strategy of player~$1$ in $G[B]$. Define the following strategy for player~$1$ in $G$:
    \[ \sigma_1(u) \EQ
    \begin{cases}
    \alpha_1(u) & \text{if } u\in V_1\cap A\\
    \sigma_1^B(u)  & \text{if } u\in V_1\cap B
    \end{cases} \]
    Let $\sigma_0$ be an arbitrary strategy of player~$0$. For every $u\in B$, $play_G^{\sigma_0,\sigma_1}(u)$ either remains forever in~$B$, or eventually enters~$A$ and remains there forever. In the first case we have $out_G^{\sigma_0,\sigma_1}(u)\ge val_{G[B]}(u)$, by the optimality of $\sigma_1^B$ in $G[B]$. In the second case, we have $out_G^{\sigma_0,\sigma_1}(u)\ge \mu_{\max}$, by Lemma~\ref{L-alpha-1}. We thus get $val_G(u)\ge \min\{ \mu_{\max},val_{G[B]}(u)\}$, for every $u\in B$.
\end{proof}

If $B\ne\emptyset$, we recursively find all the values $val_{G[B]}(u)$, for every $u\in B$, and optimal strategies $\sigma^B_0$ and $\sigma^B_1$ of the two players in~$G[B]$. (Note that this involves repeating the whole process from scratch on~$G[B]$.) The depth of the recursion is at most~$n=|V|$ as $|B|<|V|$.

Using Lemma~\ref{L-peel} we can now recover the values $val_G(u)$, for every $u\in V$. We next show how to use~$\sigma^B_0$ and~$\sigma^B_1$ to obtain optimal strategies $\sigma_0$ and $\sigma_1$ of the two players in~$G$. Let 
\begin{align*}
    L &\EQ \{ u\in B \mid val_{G[B]}(u) < \mu_{\max} \} \;, \\
    M &\EQ V\setminus L \EQ A \cup \{ u\in B \mid val_{G[B]}(u) \ge \mu_{\max} \} \;. 
\end{align*}
We let $\sigma_1$ be the strategy of player~$1$ defined in the proof of Lemma~\ref{L-peel}, and define
a strategy $\sigma_0$ for player~$0$ in~$G$ as follows:
\[
\sigma_0(u) \EQ
\begin{cases}
\alpha_0(u) & \text{if } u\in V_0\cap M\\
\sigma_0^B(u)  & \text{if } u\in V_0\cap L
\end{cases}
\]

\begin{lemma}
    The strategies $\sigma_0$ and $\sigma_1$ are optimal strategies for the two players in~$G$.
\end{lemma}

\begin{proof}
    We begin by showing that $\sigma_0$ is an optimal strategy for player~$0$ in~$G$. Let $\sigma'_1$ be an arbitrary strategy of player~$1$. Let $u\in V$ and consider $play_G^{\sigma_0,\sigma'_1}(u)$. If $u\in L$, then $\sigma_0$ ensures that the play remains in~$L$. The play is then of the form $play_{G[B]}^{\sigma^B_0,\sigma''_1}(u)$, where $\sigma''_1$ is $\sigma'_1$ restricted to~$G[B]$, and the optimality of~$\sigma^B_0$ guarantees that $out_G^{\sigma_0,\sigma'_1}(u) = out_{G[B]}^{\sigma^B_0,\sigma''_1}(u) \le val_{G[B]}(u) = val_G(u)$. If $u\in M$, there are two cases. If the play remains in~$M$, then by Lemma~\ref{L-alpha-0}, we have 
    $out_G^{\sigma_0,\sigma'_1}(u)=out_G^{\alpha_0,\sigma'_1}(u)\le \mu_{\max}=val_G(u)$. If the play ever enters~$L$, then the tail of the play remains in~$L$ and the outcome is smaller than $\mu_{\max}$.

    We next show that $\sigma_1$ is an optimal strategy for player~$1$ in~$G$. Let $\sigma'_0$ be an arbitrary strategy of player~$0$. Let $u\in V$ and consider $play_G^{\sigma'_0,\sigma_1}(u)$. We again have two cases. If $u\in A$, then by Lemma~\ref{L-alpha-1}, $out_G^{\sigma'_0,\sigma_1}(u) = out_G^{\sigma'_0,\alpha_1}(u) \ge \mu_{\max} = val_G(u)$. If $u\in B$, we have two subcases. If the play remains in~$B$, then by the optimality of~$\sigma^B_1$ we have $out_G^{\sigma'_0,\sigma_1}(u) = out_{G[B]}^{\sigma''_0,\sigma^B_1}(u) \ge val_{G[B]}(u) \ge val_G(u)$. If the play reaches~$A$ it remains in~$A$ and again by Lemma~\ref{L-alpha-1} we have $out_G^{\sigma'_0,\sigma_1}(u)\ge \mu_{\max} \ge val_G(u)$.
\end{proof}

The peeling process can also be described as an iterative process as follows. Let $G_1=G$ be the original graph. Let $\mu_1$ be the maximum value of a vertex in~$G_1$, and let $B_1\subsetneq V[G_1]$ be the set of vertices of~$G_1$ not known to have value~$\mu_1$. We then let $G_2=G_1[B_1]$, compute the maximum value~$\mu_2$ of a vertex in~$G_2$ and let~$B_2$ be the set of vertices of~$G_2$ not known to have value~$\mu_2$. We define in this way games $G_1,G_2,\ldots,G_k$, values $\mu_1,\mu_2,\ldots,\mu_k$ and sets $V=B_0\supsetneq B_1 \supsetneq\dots\supsetneq B_{k}=\emptyset$ such that $G_{i+1}=G[B_i]$. If $u\in B_{i-1}\setminus B_{i}$, then by repeated application of Lemma~\ref{L-peel} we get that $val_G(u)=\min\{\mu_1,\mu_2,\ldots,\mu_i\}$. It is interesting to note that the sequence is not necessarily increasing or decreasing.

\section{Reducing \ZeroRecover\ to \FindOther}\label{S-second}

The reduction from \ZeroRecover\ to \FindOther\ uses another simple transformation. A similar transformation is used by Bj{\"o}rklund and Vorobyov~\cite{BjVo07} in their reduction to the longest-shortest-path problem.

Let $\MPGG$ be an MPG in which all values are~$0$. We construct an MPG $G'=(V',V'_0,V_1,E',c')$ in which all values are~$0$ and in which player~$0$ has an obvious optimal strategy that ensures the $0$ values. We obtain $G'$ by adding a new sink vertex~$t$ with a self-loop $t\TO t$ of cost $c'(t\TO t)=0$. From every vertex $u\in V_0$, we add an \emph{escape} edge $u\TO t$ of cost $c'(u\TO t)=0$. We let $V'_0=V_0\cup \{t\}$. (The costs of the escape edges can actually be arbitrary, and as~$t$ has only one outgoing edge, it does not matter which player controls it.) The values of all vertices in~$G'$ are~$0$ as any use of an escape edge yields an outcome of~$0$.

\begin{lemma}\label{L-last-trans}
    The strategy $\sigma_0$ defined by $\sigma_0(u)=(u\TO t)$, for every $u\in V'_0$, is an optimal strategy for player~$0$ in~$G'$.
\end{lemma}

\begin{proof}
    Let $\sigma_1$ be an arbitrary (positional) strategy of player~$1$ in~$G$ and in~$G'$. (Note that the sets of strategies of player~$1$ in~$G$ and in~$G'$ are the same, as $V_1[G]=V_1[G']$.) Let $u\in V'$. Consider the play $play_{G'}^{\sigma_0,\sigma_1}(u)$ and the outcome $out_{G'}^{\sigma_0,\sigma_1}(u)$. The play either reaches a vertex of~$V_0$, escapes to~$t$ and remains there forever, or visits only vertices of~$V_1$. In the first case, we clearly have $out_{G'}^{\sigma_0,\sigma_1}(u)=0$. In the second case, the play is identical to a play in the original game~$G$, with respect to any strategy~$\sigma'_0$ of player~$0$ in~$G$, and thus $out_{G'}^{\sigma_0,\sigma_1}(u) = out_G^{\sigma'_0,\sigma_1}(u)\le 0$, as otherwise $\sigma_1$ could ensure a positive value from~$u$ in~$G$ against any strategy of player~$0$, contradicting the assumption that all values in~$G$ are~$0$.
\end{proof}

\begin{lemma}\label{L-GG}
    If $\sigma_1$ is an optimal strategy for player~$1$ in $G'$, then it is also an optimal strategy for player~$1$ in~$G$. 
\end{lemma}

\begin{proof}
    Since $\sigma_1$ is an optimal strategy for player~$1$ in~$G'$, and all values in~$G'$ are~$0$, we have $out_{G'}^{\sigma'_0,\sigma_1}(u)\ge 0$ for every strategy~$\sigma'_0$ of player~$0$ in~$G'$. Since every strategy $\sigma_0$ of player~$0$ in~$G$ is also a strategy of player~$0$ in~$G'$, we also have $out_G^{\sigma_0,\sigma_1}(u)\ge 0$, for every strategy $\sigma_0$ of player~$0$ in~$G$. This proves the optimality of~$\sigma_1$ in~$G$.
\end{proof}

The reduction from \ZeroRecover\ to \FindOther\ is now extremely simple. Given an MPG $G$ in which all values are~$0$, we construct the game~$G'$ and the optimal escape strategy $\sigma'_0$ for player~$0$ in~$G'$ defined in Lemma~\ref{L-last-trans}. We now call $\sigma'_1\gets\FindOther(G',\sigma'_0)$ and obtain an optimal strategy $\sigma'_1$ for player~$1$ in~$G'$. By Lemma~\ref{L-GG}, $\sigma'_1$ is also an optimal strategy of player~$1$ in~$G$.

To obtain an optimal strategy $\sigma_0$ for player~$0$ in~$G$, we do the same for the game $-G$. (Note that an optimal strategy for player~$1$ in $-G$ is an optimal strategy for player~$0$ in~$G$.)

As described, the reduction performs two calls to \FindOther. We can actually combine the two calls into one by working on a disjoint union of the escape extensions $G'$ and $(-G)'$. This gives us a Karp many-one reduction from \ZeroRecover\ to \FindOther.

\section{Reducing \Recover\ to \ZeroRecover}\label{S-ZeroRecover}

We next describe a simple Karp many-one reduction from \Recover\ to \ZeroRecover. This proves the equivalence of \ZeroRecover, \FindOther\ and \Recover\ under polynomial time Karp reductions. The only Turing reduction remaining is the main reduction from \Solve\ to \ZeroRecover.

Let $\MPGG$ be an MPG. Assume that we are given values $g(u)$, for $u\in V$, and a promise that $g(u)=val_G(u)$, for every $u\in V$. By Lemma~\ref{L-equal}, we know that if an edge $(u\to v)\in E$ is used by an optimal strategy of one of the players, then we must have $val_G(u)=val_G(v)$. This suggests constructing an MPG $G_g=(V,V_0,V_1,E_g,c_g)$ as follows:
\[ E_g \EQ \{ u\TO v \mid (u\TO v)\in E \;,\; g(u)=g(v) \} \;,\]
\[ c_g(u\TO v) \EQ c(u\TO v) - g(v) \quad,\quad (u\TO v)\in E_g \;.\]
By Lemma~\ref{L-equal}, every vertex has at least one value-preserving outgoing edge. Hence~$G_g$ is a well-defined MPG.

\begin{lemma}
    All values in~$G_g$ are~$0$. Optimal strategies $\sigma_0$ and~$\sigma_1$ of the two players in~$G_g$ are also optimal strategies of the two players in~$G$.
\end{lemma}

\begin{proof}
    Let $G'_g=(V,V_0,V_1,E_g,c)$, i.e., the game $G_g$ but with the original edge costs. By Lemma~\ref{L-equal}, we get that $val_G(u)=val_{G'_g}(u)$, for every $u\in V$, as optimal strategies for the two players in~$G$ must only use value-preserving edges. Thus, none of the players can benefit by using non-value preserving edges, i.e., edges from $E\setminus E_g$.
    
    The game $G'_g$ is composed of disjoint components corresponding to the distinct values in~$G$. The component corresponding to a value~$\mu$ is $G'_{g}[V_\mu]$, where $V_\mu=\{u\in V \mid g(u)=\mu\}$. (A component is not necessarily connected.) By subtracting $\mu$ from the costs of all edges in $G'_g[V_\mu]$, as we do in~$G_g$, the values of all the vertices in $G_g[V_\mu]$ become~$0$.

    Let $\sigma_0$ be a positional optimal strategy of player~$0$ in~$G_g$. Since~$G_g$ is a subgame of~$G$, the strategy~$\sigma_0$ is also a valid strategy of player~$0$ in~$G$. Let $\sigma_1$ be an arbitrary positional strategy of player~$1$ in~$G$. Let $u\in V$ and suppose that $g(u)=val_G(u)=\mu$. Consider $play_G^{\sigma_0,\sigma_1}(u)$. 
    
    If $play_G^{\sigma_0,\sigma_1}(u)$ only uses edges of $E_g$, then $play_G^{\sigma_0,\sigma_1}(u)$ remains in the component of~$G_g$ corresponding to~$\mu=g(u)=val_G(u)$ and $play_G^{\sigma_0,\sigma_1}(u)=play_{G_g}^{\sigma_0,\sigma_1}(u)$. By the optimality of~$\sigma_0$ in~$G_g$ we have $out_{G_g}^{\sigma_0,\sigma_1}(u)\le 0$. Since $out_{G_g}^{\sigma_0,\sigma_1}(u) =out_{G}^{\sigma_0,\sigma_1}(u)-g(u)$, we have $out_{G}^{\sigma_0,\sigma_1}(u)\le g(u)=val_G(u)$.

    If $play_G^{\sigma_0,\sigma_1}(u)$ uses edges of $E\setminus E_g$, then these edges must be from~$\sigma_1$. By Lemma~\ref{L-equal}, for any such edge $u'\TO v'$ we must have $val_G(v')<val_G(u')$. For every edge $u'\TO v'$ from~$\sigma_0$ we have $val_G(v')=val_G(u')$. Thus, all vertices in the cycle in $play_G^{\sigma_0,\sigma_1}(u)$ must have the same value $\mu'<\mu$. Since the tail of $play_G^{\sigma_0,\sigma_1}(u)$ corresponds to a play in~$G_g$ belonging to the component containing vertices of value~$\mu'<\mu$, we get by the optimality of $\sigma_0$ for~$G_g$ that $out_G^{\sigma_0,\sigma_1}(u)\le \mu'<\mu=val_G(u)$. 

    Thus, for every strategy $\sigma_1$ of player~$1$ we have $out_G^{\sigma_0,\sigma_1}(u)\le g(u)=val_G(u)$. Hence, $\sigma_0$ is also an optimal strategy for player~$0$ in~$G$. A symmetric argument shows that every optimal strategy $\sigma_1$ of player $1$ in $G_g$ is also optimal in~$G$: under $\sigma_1$ the values along a play are nondecreasing, and strictly increase whenever player $0$ uses an edge outside $E_g$.
\end{proof}

The reduction is now obvious. Given an MPG~$G$ and the values $g(u)=val_G(u)$, for $u\in V$, form the graph~$G_g$ and call $\ZeroRecover(G_g)$.

\section{Some examples}\label{S-examples}

We now give some simple examples that illustrate some points made earlier in the paper.

Figure~\ref{fig:counter} provides a simple example that shows that an optimal counter strategy of an optimal strategy is not necessarily an optimal strategy.

Figure~\ref{fig:example1} provides two simple games~$G$ and~$H$ such that both players have entangled, i.e., non-product, optimal strategies in $G\times H$.

A second example of this form is given in Figure~\ref{fig:example2}. This example has two added features. The first is that when both players play according to their entangled optimal strategies $\rho_0$ and $\rho_1$, the same cycle is eventually entered, from every initial vertex. The second advantage is that the optimality of $\rho_0$ and $\rho_1$ can be certified via Theorem~\ref{T-gain-bias} using the biases shown in green in the figure.

\begin{figure}[t]
\centering
    \includegraphics[width=0.4\linewidth]{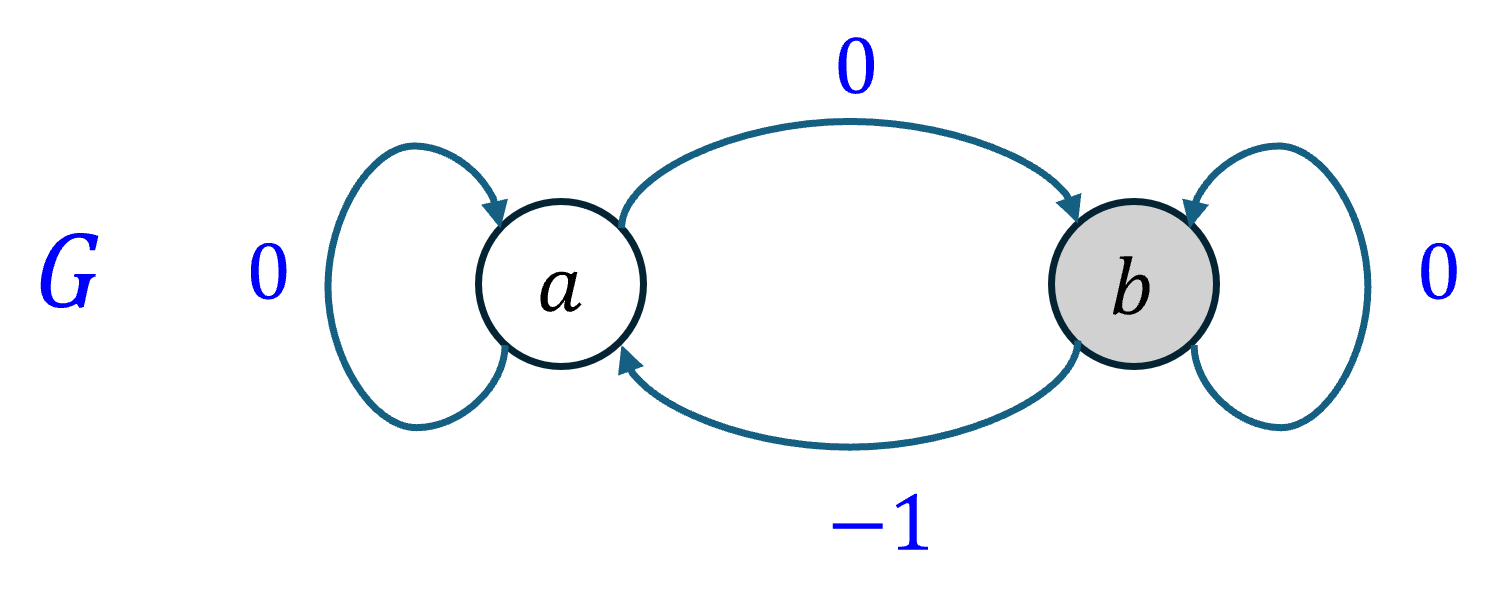}
    \caption{An optimal counter strategy to an optimal strategy is not necessarily an optimal strategy. In the figure, $V_0=\{a\}$ and $V_1=\{b\}$. The strategy $\sigma_0(a)=(a\TO a)$ is an optimal strategy of player~$0$. The strategy $\sigma'_0(a)=(a\TO b)$ is another optimal strategy of player~$0$. They both ensure values of at most~$0$ from both~$a$ and~$b$. The strategy $\sigma_1(b)=(b\TO a)$ is an optimal counter strategy to~$\sigma_0$, but not to~$\sigma'_0$. Hence it is not an optimal strategy for player~$1$. The optimal strategy for player~$1$ is $\sigma'_1(b)=(b\TO b)$ which ensures a value of~$0$ from both vertices.}
    \label{fig:counter}
\end{figure}

\begin{figure}[p]
\centering
\begin{minipage}[c]{0.35\textwidth}
    \centering
    \includegraphics[width=\linewidth]{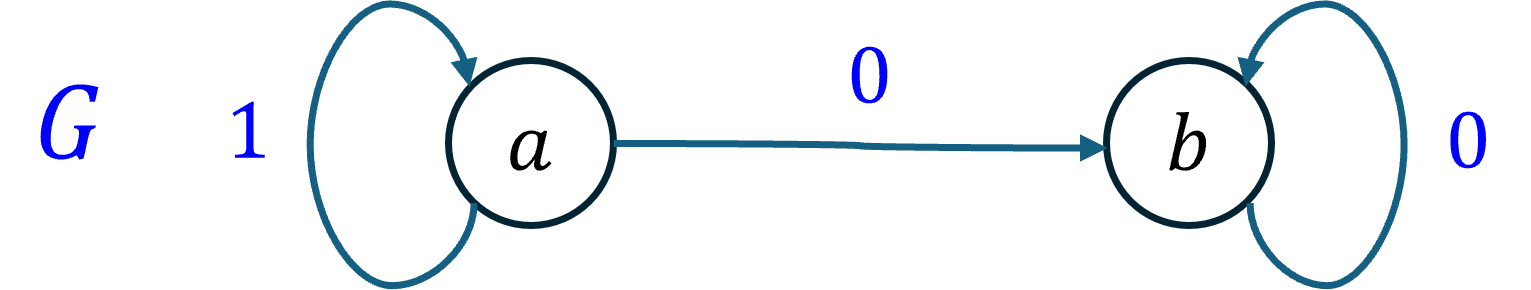}

    \vspace{1cm}

    \includegraphics[width=\linewidth]{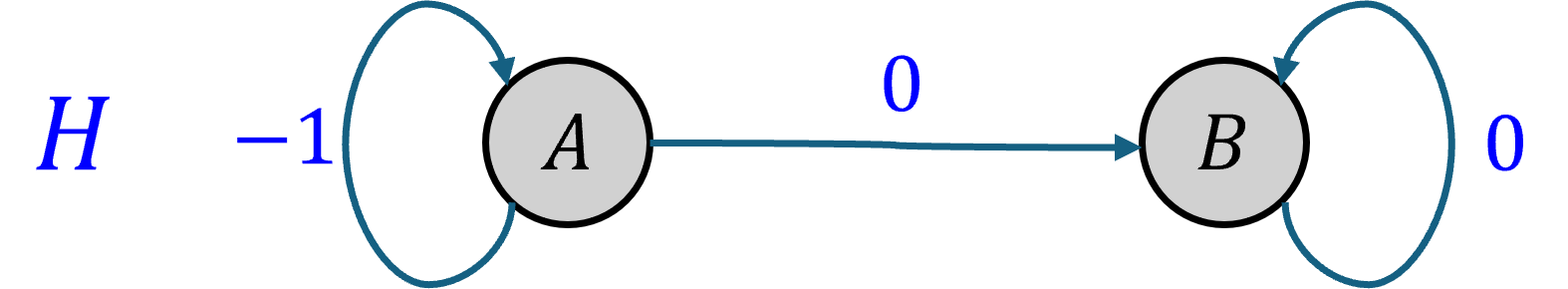}
\end{minipage}
\hspace*{1cm}
\begin{minipage}[c]{0.4\textwidth}
    \centering
    \includegraphics[width=\linewidth]{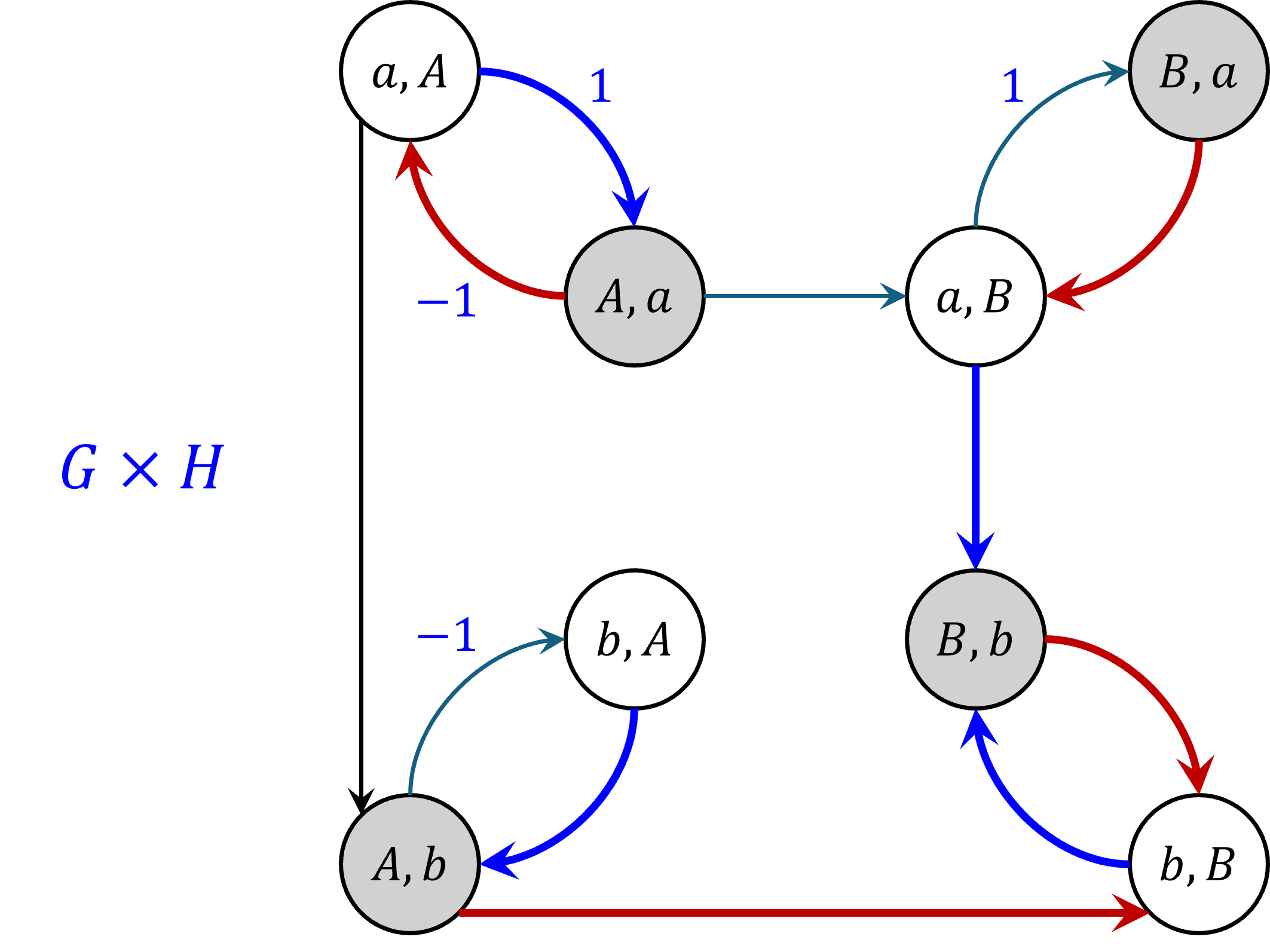}
\end{minipage}
    \caption{An example of two simple games~$G$ and~$H$ such that in $G\times H$ both players have entangled optimal strategies, i.e., optimal strategies that are not product strategies. The vertices $a$ and~$b$ of~$G$ are controlled by player~$0$. The vertices $A$ and~$B$ of~$H$ are controlled by player~$1$. All omitted edge costs in $G\times H$ are~$0$. All values in~$G$ and in~$H$, and hence also in $G\times H$, are~$0$. An optimal strategy $\rho_0$ of player~$0$ is given by the bold blue edges. The strategy is entangled as $\rho_0(a,A)=((a,A)\TO (A,a))$ while $\rho_0(a,B)=((a,B)\TO (B,b))$. Similarly, an optimal strategy $\rho_1$ of player~$1$ is given by the bold brown edges. The strategy is entangled as $\rho_1(A,a)=((A,a)\TO (a,A))$ while $\rho_1(A,b)=((A,b)\TO (b,B))$.}
    \label{fig:example1}
\end{figure}

\begin{figure}[p]
\centering
\begin{minipage}[c]{0.35\textwidth}
    \centering
    \includegraphics[width=\linewidth]{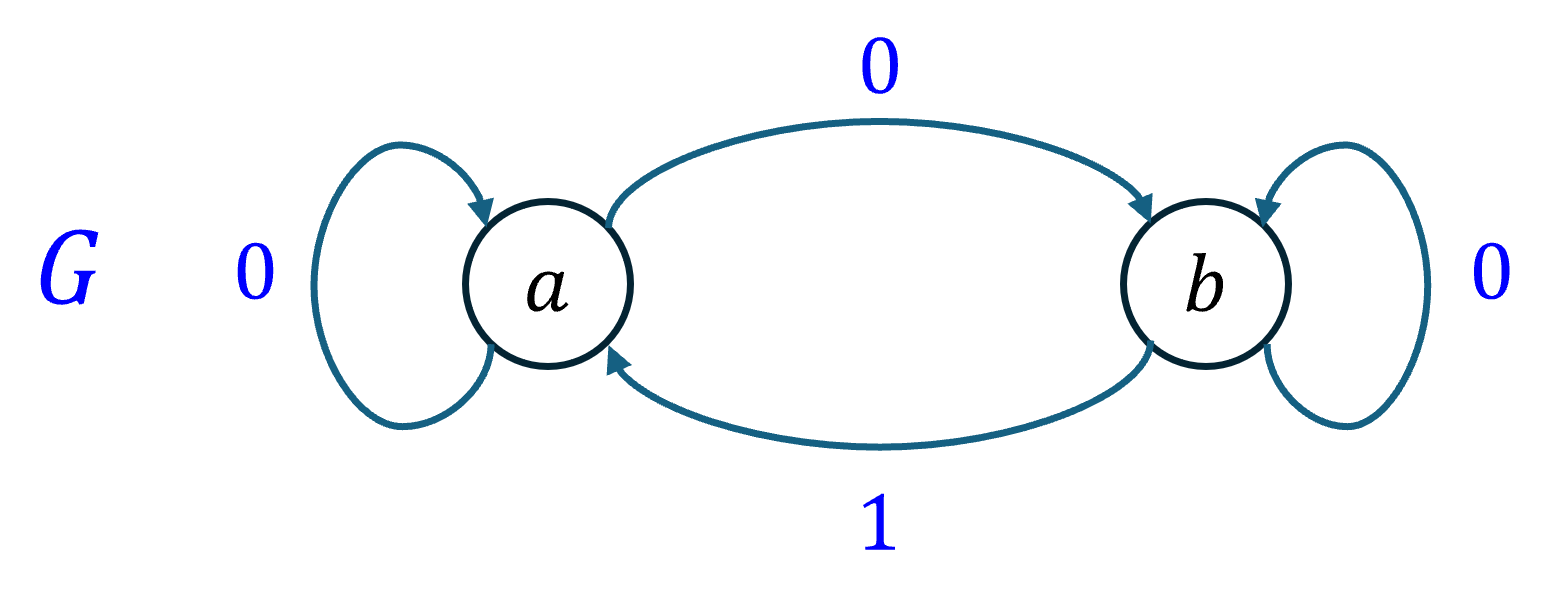}

    \vspace{1cm}

    \includegraphics[width=\linewidth]{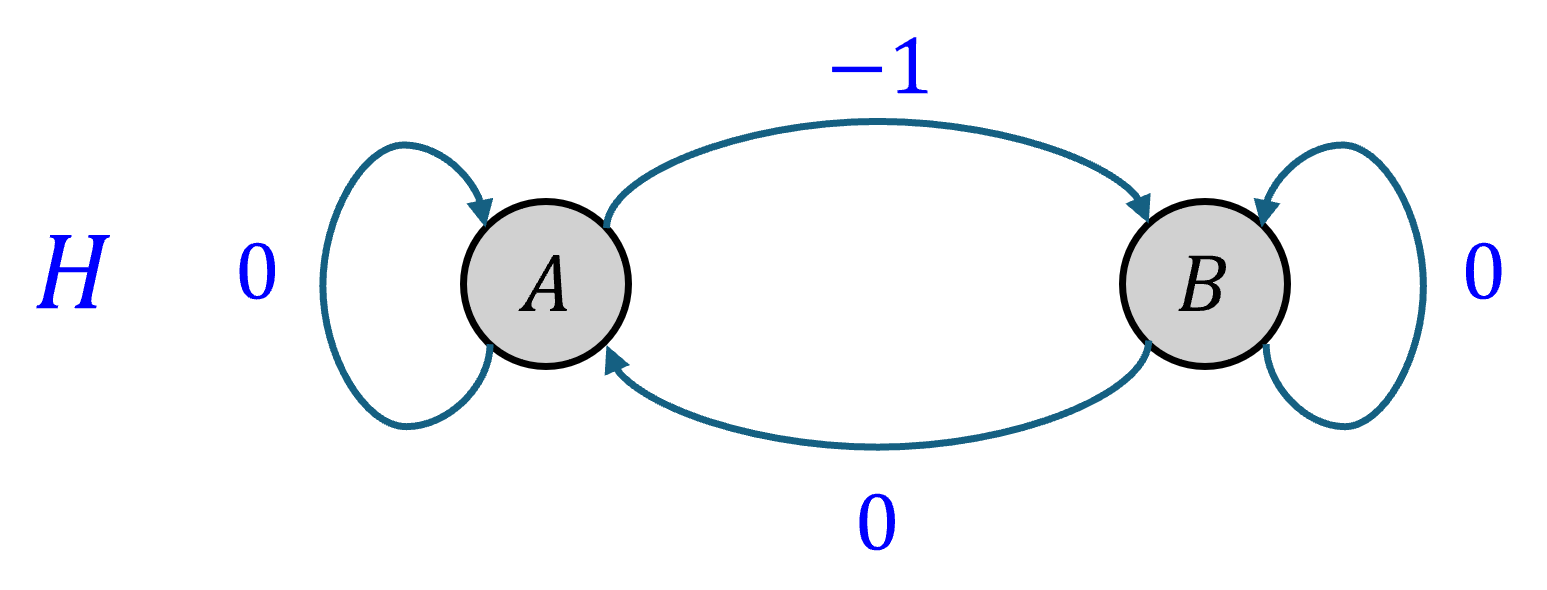}
\end{minipage}
\hspace*{1cm}
\begin{minipage}[c]{0.4\textwidth}
    \centering
    \includegraphics[width=\linewidth]{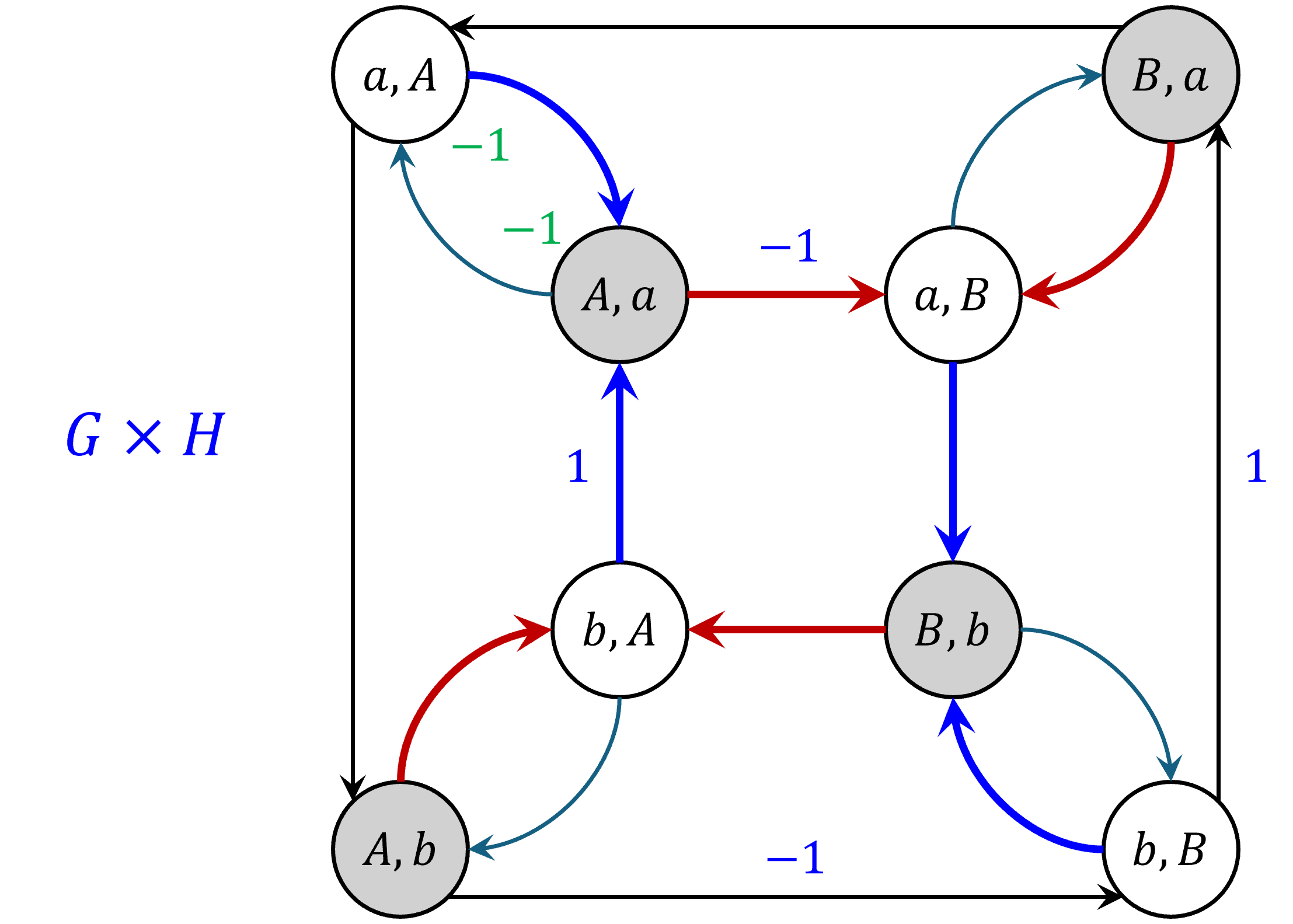}
\end{minipage}
    \caption{The second example. The entangled optimal strategies $\rho_0$ and $\rho_1$ of the two players in $G\times H$ are again given by the bold blue and brown edges. All values are again~$0$. The optimality of $\rho_0$ and $\rho_1$ can be certified using the biases $b(a,A)=b(A,a)=-1$, and 0 for all other pairs. The two non-zero biases are shown in green next to $(a,A)$ and $(A,a)$.}
    \label{fig:example2}
\end{figure}

\section{Concluding remarks and open problems}\label{S-concl}

We showed that the four problems \Solve, \ZeroRecover, \FindOther\ and \Recover\ are equivalent under polynomial time Turing reductions. Furthermore, the last three problems are equivalent under polynomial time Karp reductions. Are the four problems equivalent under polynomial time Karp reductions?

As pointed out by a reviewer, we can actually add a fifth problem to the list of equivalent problems. Let \mbox{\sc Value} be the problem of computing the values of all vertices of a game. A simple self-reducibility argument shows that \Solve\ Turing-reduces to \mbox{\sc Value}.

The most interesting open problem, of course, is to obtain improved algorithms for the solution of MPGs. The techniques that led to a \emph{quasipolynomial} time algorithm for \emph{Parity Games} do not seem to give significant asymptotic improvements for MPGs. (See, e.g., Fijalkow, Gawrychowski, and Ohlmann \cite{FijalkowGO20}, Colcombet, Fijalkow, Gawrychowski, and Ohlmann \cite{ColcombetFGO22}.) Thus, new techniques are probably needed to obtain improved algorithms.

\subsection*{Note added in proof} OpenAI \cite{OAI:Deterministic-quasipolynomial-time-mean-payoff-games-September-25-2026} recently announced a quasipolynomial-time algorithm for solving MPGs. As this is not yet a polynomial-time algorithm, all the reductions in this paper are still meaningful. 

\section*{Declaration} 

Interactions with ChatGPT 5.6 Pro influenced the content of Section~\ref{sub-peeling}, Section~\ref{S-second} and Section~\ref{S-ZeroRecover}. The authors assume responsibility for all content.

\section*{Acknowledgments}

The second author would like to thank Peter Bro Miltersen for bringing the strategy recovery problem to his attention and for some interesting discussions. We also thank an anonymous reviewer for pointing out that a fifth problem can be added to the list of equivalent problems.

\bibliography{references}

\end{document}